%% file: main.tex
\documentclass{article}

\usepackage{times}

\usepackage{amsmath, amssymb, amsthm, amsfonts}
\usepackage{algorithm}
\usepackage{algorithmic}
\usepackage{fullpage}
\usepackage{thmtools, thm-restate}
\usepackage{bbm}
\usepackage{xcolor,verbatim,hyperref}
\usepackage{natbib}

\newtheorem{definition}{Definition}
\newtheorem{theorem}{Theorem}
\newtheorem{lemma}{Lemma}
\newtheorem{assumption}{Assumption}

\newtheorem*{cor*}{Corollary 1 (Main result (Formal))}

\newcommand{\rc}[1]{\textcolor{red}{[Rachel: #1]}}

\newcommand{\todo}[1]{\textcolor{DarkGreen}{[TODO: #1]}}

\newcommand{\R}{\mathbb{R}}
\newcommand{\E}{\mathbb{E}}

\newcommand{\cL}{\mathcal{L}}
\newcommand{\cC}{\mathcal{C}}
\newcommand{\cD}{\mathcal{D}}
\newcommand{\cO}{\mathcal{O}}
\newcommand{\cR}{\mathcal{R}}

\newcommand{\cM}{\mathcal{M}}
\newcommand{\eps}{\epsilon}
\newcommand{\tily}{\hat{y}}
\newcommand{\hatth}{\hat{\theta}}

\newcommand{\rr}{{\hat{\theta}^{R}}}
\newcommand{\lr}{{\hat{\theta}^{L}}}

\newcommand{\pr}{{\hat{\theta}^{P}}}

\newcommand{\argmin}{\mathop{\arg\min}}
\newcommand{\supp}{\mathop{\mathtt{supp}}}
\newcommand{\tr}{\mathop{\mathtt{trace}}}
\newcommand{\rank}{\mathop{\mathrm{rank}}}
\newcommand{\Cov}{\mathop{\mathtt{Cov}}}
\newcommand{\bias}{\mathop{\mathtt{bias}}}

\title{Truthful Linear Regression\thanks{The first author was funded in part by NSF grant CNS-1254169, US-Israel Binational Science Foundation grant 2012348, and a Google Faculty Research Award.  The third author was funded in part by NSF grant CNS-1254169, US-Israel Binational Science Foundation grant 2012348, the Charles Lee Powell Foundation, a Google Faculty Research Award, an Okawa Foundation Research Grant, and a Microsoft Faculty Fellowship.  Work completed in part while the first and second authors were at Technicolor Research Labs.  We thank Jenn Wortman Vaughan for her comments on the final version of this paper.}}
\author{Rachel Cummings\thanks{Computing and Mathematical Sciences, California Institute of Technology;
\texttt{rachelc@caltech.edu}.} \and Stratis Ioannidis\thanks{Yahoo! Labs;
\texttt{stratisioannidis@yahoo-inc.com}} \and Katrina Ligett\thanks{Computing and Mathematical Sciences, California Institute of Technology;
\texttt{katrina@caltech.edu}} }

\begin{document}
\maketitle

\begin{abstract}
\input{abstract}
\end{abstract}

 \vfill
 \thispagestyle{empty}
\setcounter{page}{0}
\pagebreak

\input{intro.tex}

\input{prelim.tex}

\input{nonprivate.tex}

\input{privacy.tex}

\input{analysis.tex}



\bibliographystyle{plain}
\bibliography{./refs}

\appendix

\input{prelim-appendix}
\input{appendixproofs}

\input{appendix.tex}

\input{appendix2}

\end{document}

%% file: abstract.tex
We consider the problem of fitting a linear model to data held by individuals who are concerned about their privacy. Incentivizing most players to truthfully report their data to the analyst constrains our design to mechanisms that provide a privacy guarantee to the participants; we use differential privacy to model individuals' privacy losses. This immediately poses a problem, as differentially private computation of a linear model necessarily produces a biased estimation, and existing approaches to design mechanisms to elicit data from privacy-sensitive individuals do not generalize well to biased estimators. We overcome this challenge through an appropriate design of the computation and payment scheme.

%% file: intro.tex
\section{Introduction}

Fitting a linear model is perhaps the most fundamental and basic learning task, with diverse applications from statistics to experimental sciences like medicine and sociology. In many settings, the data from which a model is to be learnt are not held by the analyst performing the regression task,  but must be elicited from individuals. Such settings clearly include medical trials and census surveys, as well as mining  online behavioral data, a practice currently happening  at a massive scale.  

If data are held by self-interested individuals, it is not enough to simply run a regression---the data holders may wish to influence the outcome of the computation, either because they could benefit directly from certain outcomes, or to mask their input due to privacy concerns. In this case, it is necessary to model the utility functions of the individuals and to design mechanisms that provide proper incentives. Ideally, such mechanisms should still allow for accurate computation of the underlying regression. A tradeoff then emerges between the accuracy of the computation and the budget required to compensate participants. 

In this paper, we focus on the problem posed by data holders who are concerned with their privacy. Our approach can easily be generalized to handle individuals manipulating the computation's outcome for other reasons, but for clarity we treat only privacy concerns. We consider a population of players, each holding private data, and an analyst who wishes to compute a linear model from their data. The analyst must design a mechanism (a computation he will do and payments he will give the players) that incentivizes the players to provide information that will allow for accurate computation, while minimizing the payments the analyst must make.

We use a model of players' costs for privacy that is based on the well-established notion of differential privacy~\citep{DMNS06}. Incentivizing most players to truthfully report their data to the analyst constrains our design to mechanisms that are differentially private. This immediately poses a problem, as differentially private computation of a linear model necessarily produces a biased estimation; existing approaches~\citep{GLRS14} to design mechanisms to elicit data from privacy-sensitive individuals do not generalize well to biased estimators. Overcoming this challenge, through appropriate design of the computation and payment scheme, is the main technical contribution of the present work. 

\subsection{Our Results}
We study the above issues in the context of linear regression. We present a mechanism (Algorithm~\ref{alg:private}), which, under appropriate choice of parameters and fairly mild technical assumptions, satisfies the following properties: it is (a) \emph{accurate} (Theorem~\ref{thm.accuracy}), i.e., computes an estimator whose squared $L_2$ distance to the true linear model goes to zero as the number of individuals increases, (b) \emph{asymptotically truthful} (Theorem~\ref{thm.equil}), in that agents have no incentive to misreport their data, (c) it \emph{incentivizes participation} (Theorem~\ref{thm.privateir}), as players receive positive utility, and (d) it requires an \emph{asymptotically small budget} (Theorem~\ref{thm.budget}),  as total payments to agents go to zero as the number of individuals increases.  Our technical assumptions are on how individuals experience privacy losses and on the distribution from which these losses are drawn.  Accuracy of the computation is attained by establishing that the algorithm provides differential privacy (Theorem~\ref{thm.priv}), and that it provides payments such that the vast majority of individuals are incentivized to participate and to report truthfully (Theorems~\ref{thm.equil} and~\ref{thm.privateir}). An informal statement appears in Theorem~\ref{thm.informal}.

The fact that our total budget decreases in the number of individuals in the population is an effect of the approach we use to eliciting truthful participation, which is based on the peer prediction technology (Appendix~\ref{Brier}) and of the model of agents' costs for privacy (Section~\ref{s.costs}). A similar effect was seen by \cite{GLRS14}. As they note, costs would no longer tend to zero if our model incorporated some fixed cost for interacting with each individual.

\subsection{Related Work}

Following \cite{GR11}, a series of papers have studied data acquisition problems from agents that have privacy concerns. The vast majority of this work \citep{FL12,LR12,NVX14, CLR+15} operates in a model where agents cannot lie about their private information (their only recourse is to withhold it or perhaps to lie about their costs for privacy). A related thread~\citep{GR11,NOS12,Harvardetal} explores cost models based on the notion of differential privacy~\citep{DMNS06}.

Our setting is closest to, and inspired by, \cite{GLRS14}, who bring the technology of peer prediction to bear on the problem of incentivizing truthful reporting in the presence of privacy concerns. The peer prediction approach of \cite{MRZ05} incentivizes truthful reporting (in the absence of privacy constraints) by rewarding players for reporting information that is predictive of the reports of other agents. This allows the analyst to leverage correlations between players' information. \cite{GLRS14} adapt the peer prediction approach to overcome a number of challenges presented by privacy-sensitive individuals.
 The mechanism and analysis of~\cite{GLRS14} was for the simplest possible statistic---the sum of private binary types. In contrast, we regress a linear model over player data, a significantly more sophisticated learning task. In particular, to attain accurate, privacy-preserving linear regression, we deal with biased  private estimators, which interferes with our ability to incentivize truth-telling, and hence to compute an accurate statistic.

Linear regression under strategic agents has been studied in a variety of different contexts. \cite{DFP10} consider an analyst that regresses a ``consensus'' model across data coming from multiple strategic agents; agents would like the consensus value to minimize a loss over their own data, and they show that, in this setting, empirical risk minimization is group-strategyproof. A similar result, albeit in a more restricted setting, is established by \cite{PP04}. Regressing a linear model over data from strategic agents that can only manipulate their costs, but not their data, was studied by \cite{HIM14} and \cite{CDP14}, while \cite{IL13} consider a setting without payments, in which agents receive a utility as a function of estimation accuracy. We depart from the above approaches by considering agents whose utilities depend on their loss of {\em privacy}, an aspect absent from the above works. 

Finally, we note a growing body of work on differentially private empirical risk minimization.  Our mechanism is based on the outcome perturbation algorithm of \cite{CMS11}.  Other algorithms from this literature --- such as the localization algorithm of \cite{BST14} or objective perturbation of \cite{CMS11} --- could be used instead, and would likely yield even better accuracy guarantees.  We chose the output perturbation mechanism because it provides an explicit characterization of the noise added to preserve privacy, which allows the analysis to better highlight the challenges of incorporating privacy into our setting.


%% file: prelim.tex
\section{Model and Preliminaries}

We present our model and a technical preliminary in this section. A more detailed review of peer prediction, linear regression, and differential privacy can be found 
in Appendix~\ref{sec:prelim-appendix}. 

\subsection{A Regression Setting}
We consider a population where each player $i\in [n]\equiv \{1,\ldots,n\}$ is associated with a vector  $x_i \in \R^d$ (i.e., player $i$'s \emph{features}) and a variable $y_i \in \R$ (i.e., her \emph{response} variable). 
We assume that responses  are linearly related to the features; that is, there exists a $\theta\in\R^d$ such that  
\begin{align}y_i = \theta^{\top} x_i + z_i,\quad \text{for all}~i\in[n],\label{linear} \end{align}
where $z_i$ are zero-mean noise variables. 

An analyst wishes to infer a linear model from the players' data; that is, 
he wishes to estimate $\theta$, e.g., by performing linear regression on the players' data. However, players incur a privacy cost from revelation of their data and need to be properly incentivized to truthfully reveal it to the analyst. More specifically,
 as in \cite{IL13}, we assume that player $i$ can manipulate  her responses $y_i$ \emph{but not} her features $x_i$. This is indeed the case when features are measured directly by the analyst (e.g., are observed during a physical examination or are measured in a lab)  or are verifiable (e.g., features are extracted from a player's medical record or are listed on her ID). A player may misreport her response $y_i$, on the other hand, which is unverifiable; this would be the case if, e.g., $y_i$ is the answer the player gives to a survey question about her preferences or habits.  
 
 We assume that players are strategic and may lie either to increase the payment they extract from the analyst or to mitigate any privacy violation they incur by the disclosure of their data.
To address such strategic behavior, the analyst will design a mechanism $\cM:(\R^d\times \R)^n \to \R^d\times \R_+^n$ that takes as input all player data (namely, the features $x_i$ and possibly perturbed responses $\tily_i$), and outputs an estimate $\hatth$ and a set of non-negative payments $\{\pi_i\}_{i\in [n]}$ to each player. Informally, we seek mechanisms that allow for {\em accurate} estimation of $\theta$ while requiring only asymptotically {\em small budget}. In order to ensure accurate estimation of $\theta$, we will require that our mechanism {\em incentivizes truthful participation} on the part of most players, which in turn will require that we provide an appropriate {\em privacy guarantee}. We discuss privacy in more detail in Section~\ref{sec:jointdiff}.
Clearly, all of the above also depend on the players' rational behavior and, in particular, their utilities; we formally present our model of player utilities in Section~\ref{s.costs}.

Throughout our analysis, we assume that  $\theta$ is drawn independently from a known distribution $\mathcal{F}$, the attribute vectors $x_i$ are drawn independently from the uniform distribution on the $d$-dimensional unit ball,\footnote{See Theorem \ref{thm.converge} and its accompanying Remark in Appendix \ref{sec:linearreg} for a discussion of generalizing beyond the uniform distribution.}
 and the noise terms $z_i$ are drawn independently from a known distribution $\mathcal{G}$.  Thus $\theta$, $\{x_i\}_{i\in[n]}$, and $\{z_i\}_{i\in [n]}$ are independent random variables, while responses $\{y_i\}_{i\in[n]}$ are determined by \eqref{linear}.   Note that as a result, responses are conditionally independent given $\theta$. 

We require some additional bounded support assumptions on these distributions.
In short, these boundedness assumptions are needed to ensure the sensitivity of mechanism $\cM$ is finite; it is also natural in practice that both features and responses take values in a bounded domain.
More precisely, we assume that the distribution $\mathcal{F}$ has  bounded support, such that $\left\|\theta\right\|_2^2 \leq B$ for some constant $B$; we also require the noise distribution $\mathcal{G}$ to have mean zero, finite variance $\sigma^2$, and bounded support: $\supp(\mathcal{G})=[-M,M]$ for some constant $M$.  These assumptions together imply that $\left|\theta^{\top}x_i\right| \leq B$ and $\left|y_i\right| \leq B+M$.   

\subsection{Linear and Ridge Regression}\label{s.linreg}
Let $X=[x_i]_{i\in [n]}\in \R^{n\times d}$ denote the $n\times d$ matrix of features, and $y=[y_i]_{i\in [n]}\in \R^n$ the vector of responses. Estimating $\theta$ through \emph{ridge regression} amounts to minimizing the following regularized quadratic loss function: 
\begin{align}\cL(\theta; X,y) = \sum_{i=1}^n \ell(\theta; x_i, y_i) = \sum_{i=1}^n (y_i - \theta^{\top} x_i)^2 + \gamma \left\| \theta \right\|_2^2. \label{regularizedloss}\end{align}
That is, the ridge regression estimator can be written as:
$ \rr = \argmin_{\theta \in \R^d} \cL(\theta; X,y) = (\gamma I + X^{\top} X)^{-1} X^{\top} y. $
The parameter $\gamma>0$, known as the regularization parameter, ensures that the loss function is \emph{strongly convex} (see Appendix~\ref{app:strongconv}) and, in particular, that the minimizer of \eqref{regularizedloss} is unique. When $\gamma=0$, the estimator is the standard \emph{linear regression} estimator, which we denote by $\lr = (X^{\top} X)^{-1} X^{\top} y $. The linear regression estimator is unbiased, i.e., under \eqref{linear}, it satisfies $\E[\lr] = \theta$.   The same is not true when $\gamma>0$; the general ridge regression estimator $\rr$ is \emph{biased}.  

\subsection{Differential Privacy}\label{sec:jointdiff}

 Recall the classic definition of differential privacy by~\cite{DMNS06}:
\begin{definition}[Differential Privacy \citep{DMNS06}]
A mechanism $\cM: \mathcal{D}^n \rightarrow \mathcal{R}$ is \emph{$\epsilon$-differentially private} if for every pair of  databases $D, D' \in \mathcal{D}^n$ differing only in one element, and for every subset of possible
outputs $\mathcal{S} \subseteq \mathcal{R}$,
$ \Pr[\cM(D) \in \mathcal{S}] \leq \exp(\epsilon)\Pr[\cM(D') \in \mathcal{S}]. $
\end{definition}

We depart from this classic definition, quantifying privacy violation instead through \emph{joint differential privacy} \citep{KPRU14}.
 Intuitively, full differential privacy  requires that all outputs by the mechanism $\cM$, including the payment it allocates to a player, is insensitive to every player's input. In settings like ours, however, it makes sense to assume that the payment to a player is also in some sense  ``private,'' in that it is shared neither publicly nor with other players. To that end, we assume that the estimate $\hatth$ computed by the mechanism $\cM$ is a publicly observable output; in contrast, each payment $\pi_i$ is observable \emph{only by player} $i$. Hence, from the perspective of each player $i$, the mechanism output that is publicly released and that, in turn, might violate her privacy, is $(\hatth,\pi_{-i})$, where $\pi_{-i}$ comprises all payments excluding player $i$'s payment. 

\begin{definition}[Joint Differential Privacy \citep{KPRU14}]\label{def:jointdiff}
Consider a mechanism $\cM:\cD^n\to \cO\times \cR^n$, for $\cD,\cO,\cR$ arbitrary sets. For each $i\in[n]$, let $\left(\cM(\cdot)\right)_{-i} = (o, \pi_{-i})\in \cO\times\cR^{n-1}$ denote the portion of the mechanism's
output that is observable to outside observers and players $j \neq i$. A mechanism $\cM$ is $\epsilon$-jointly differentially private if, for every player $i$, every database $D\in\cD^n$, every $d'_i\in \cD$, 
and for every observable set of outcomes $\mathcal{S} \subseteq \cO\times \cR^{n-1}$:
\[ \Pr\left[\left(\cM(D)\right)_{-i} \in \mathcal{S}\right] \leq \exp(\epsilon)\Pr\left[\left(\cM(d'_i,D_{-i})\right)_{-i} \in \mathcal{S}\right]. \]
\end{definition}
This relaxation of differential privacy is  natural, but it is also necessary to incentivize truthfulness. Requiring that a player's payment $\pi_i$ be $\epsilon$-differentially private implies that a player's unilateral deviation changes the distribution of her payment only slightly. Hence, under full differential privacy, a player's payment would remain roughly the same no matter what she reports, which intuitively cannot incentivize truthful reporting. 


We emphasize here that the existence of priors and the independence of responses are used only to prove the accuracy of the model learned and truthfulness, but not to ensure any privacy guarantee.  Our mechanism satisfies joint differential privacy regardless of of whether the assumptions hold; if they do, accuracy and truthfulness follow.  Further, the notion of $\epsilon$-joint differential privacy depends on both $y_i$ \emph{and} $x_i$: although a player can only manipulate $y_i$, both her response \emph{and} her features are treated as ``private'' variables in our model, and both disclosures incur a privacy cost. Features should certainly be deemed private if, e.g., they are attributes in a player's medical record, or outcomes of a medical examination. Moreover, \eqref{linear} implies a correlation between features and the response, which can be strong, for example, in the case where $\theta$ has small support; it is therefore reasonable to assume that, if the response is private, so should features correlated to this response.

\subsection{Player Utilities}\label{s.costs}


As discussed in the related work section, starting from \cite{GR11}, a series of recent papers  on strategic data revelation model player  privacy costs as functions of the privacy parameter $\epsilon$.  We also adopt this modeling assumption.
Having introduced the notion of joint differential privacy, we now present our model of player utilities. We assume that every player is characterized by a cost parameter $c_i\in \R_+$, determining her sensitivity to the privacy violation incurred by the revelation of her data to her analyst. In particular, 
each player has a privacy cost function $f_i(c_i, \eps)$ that describes the cost she incurs when her data is used in an $\eps$-jointly differentially private computation.  Players have quasilinear utilities, so if player $i$ receives payment $\pi_i$ for her report, and experiences cost $f_i(c_i, \eps)$ from her privacy loss, her utility is
$ u_i = \pi_i - f_i(c_i, \eps). $


Following again recent work, 
we assume that $f_i$ can be an arbitrary function, bounded by an increasing monomial of $\epsilon$.  In particular, we make the following assumption.
\begin{assumption}\label{a.costs}
The privacy cost function of each player satisfies
$f_i(c_i,\eps) \leq c_i \epsilon^2.$ 
\end{assumption}

The monotonicity in $\epsilon$ is intuitive, as smaller values imply stronger privacy properties, with $\epsilon=0$ indicating the output is independent of player $i$'s data. We note that the quadratic bound in Assumption \ref{a.costs} was introduced by \cite{Harvardetal} and also adopted by \cite{GLRS14}. 
As noted by the above authors, the quadratic bound can be shown to hold for a broad class of natural cost functions $f_i$; we refer the reader to Appendix~\ref{app:quad} for a formal description of this class.

Throughout our analysis, we assume that the privacy cost parameters are also random variables, sampled from a distribution $\cC$. We allow $c_i$ to depend on  player $i$'s data $(x_i, y_i)$; however, we assume conditioned on $(x_i, y_i)$, that $c_i$ does not reveal any additional information about the costs or data of any other agents.    Formally: 
\begin{assumption}\label{a.indep}
Given $(x_i, y_i)$, $(X_{-i}, y_{-i},c_{-i})$  is conditionally independent of $c_i$:
\[ \Pr[ (X_{-i}, y_{-i},c_{-i}) | (x_i, y_i), c_i ] = \Pr[ (X_{-i}, y_{-i},c_{-i}) | (x_i, y_i), c_i' ] \text{ for all } (X_{-i}, y_{-i},c_{-i}), \; (x_i, y_i), \; c_i, \; c_i' .\]
\end{assumption}

We also make the following additional technical assumption on the tail of $\cC$.
\begin{assumption}\label{a.tail}
The conditional marginal distribution satisfies $ \min_{x_i,y_i} \left( \Pr_{c_j \sim \mathcal{C}|x_i, y_i} [c_j \leq \tau] \right) \geq 1 - \tau^{-p}$ for some constant $p>1$.  \end{assumption}
Note that Assumption~\ref{a.tail} implies that $\Pr_{c_i \sim \mathcal{C}} [c_i \leq \tau] \geq 1 - \tau^{-p}$. 

\input{mechprop.tex}

%% file: mechprop.tex
\subsection{Mechanism Properties}
We seek mechanisms that satisfy the following properties: (a) truthful reporting is an equilibrium, (b) the estimator computed under truthful reporting is highly accurate, (c) players are ensured non-negative utilities from truthful reporting,  and (d) the budget required from the analyst to run the mechanism is small.
We present here the standard definitions for these properties used in this paper.  Consider a regression mechanism $\cM$.  Let $\pi_i(X,y)$ and be the payment to player $i$ when $(X,y)$ is the collection of reports to the regression mechanism, and let $f_i(c_i, \eps)$ be player $i$'s cost for participating in the mechanism. 
We define a strategy profile $\sigma = (\sigma_1, \ldots, \sigma_n)$ to be a collection of strategies $\sigma_i$ (one for each player), mapping from realized data $(x_i, y_i)$ to reports $\tily_i$.  Under strategy $\sigma_i$, a player who has data $(x_i, y_i)$ would report $\tily_i = \sigma_i(x_i, y_i)$ to the regression mechanism.

\begin{definition}[Bayes Nash equilibrium]
A strategy profile $\sigma$ forms an \emph{$\eta$-approximate Bayes Nash equilibrium} if for every player $i$, for all realizable $(x_i,y_i)$, and for every misreport $\tily_i \neq y_i$,
\[ \E[\pi_i(X,\sigma(X,y))] - f_i(c_i, \eps) \geq \E[\pi_i(X,(\tily_i,\sigma_{-i}(X_{-i},y_{-i})))] - f(c_i, \eps) - \eta. \]
\end{definition}

\begin{definition}[Accuracy]
A regression is \emph{$\eta$-accurate} if for all realizable parameters $\theta$, it outputs an estimate $\hatth$ such that
$\E[ \| \hatth - \theta \|_2^2 ] \leq \eta. $
\end{definition}


\begin{definition}[Individually Rational]
A mechanism is \emph{individually rational} (IR) if  $\; \E[\pi_i(X,y)] - f_i(c_i, \eps) \geq 0  $  for every player $i$ and for all realizable $(X, y)$.

\end{definition}

We will also be concerned with the total amount spent by the analyst in the mechanism.  The \emph{budget} $\mathcal{B}$ of a mechanism is the sum of all payments made to players.  That is, $\mathcal{B} = \sum_i \pi_i$.

\begin{definition}[Asymptotically small budget]
An \emph{asymptotically small budget} is such that $\mathcal{B} = \sum_{i=1}^n \pi_i(X,y) = o(1),$  
for all realizable $(X,y)$.

\end{definition}


%% file: nonprivate.tex
\section{Truthful Regression  without Privacy Constraints}\label{s.regression}

To illustrate the ideas we  use in the rest of the paper, we present in this section a mechanism which incentivizes truthful reporting in the absence of privacy concerns.  If the players do not have privacy concerns (i.e., $c_i=0$ for all $i\in [n]$), the analyst can simply collect data, estimate $\theta$ using linear regression, and compensate players using the following scoring rule:\footnote{This is a variant of the well-known Brier scoring rule \citep{Bri50}.  See Appendix~\ref{Brier} for more details.} 
\[ B_{a,b}(p,q) = a - b \left(p - 2pq + q^2\right). \]
The mechanism is formally presented in Algorithm \ref{alg:formal}. In the spirit of peer prediction, a player's payment depends on how well her reported $\tily_i$ agrees with the predicted value of $y_i$, as constructed by the estimate $\lr_{-i}$ of $\theta$  produced by all her peers.  
We now show that truthful reporting is a Bayes Nash equilibrium.
\begin{algorithm}[!t]
  \caption{Truthful Regression Mechanism($a$, $b$)}
  \label{alg:formal}
  \begin{algorithmic}
    \STATE{Solicit reports $X \in (\R^d)^n$ and $\tily \in \mathbb{R}^n$}
    \STATE{Analyst computes $\lr = (X^{\top} X)^{-1}X^{\top} \tily$ and $\lr_{-i} = (X_{-i}^{\top} X_{-i})^{-1}X_{-i}^{\top} \tily_{-i}$ for each $i \in [n]$}
    \STATE{Output estimator $\lr$}
    \STATE{Pay each player $i$, $\pi_i = B_{a,b}( x_i^{\top}\lr_{-i},x_i^{\top}\E[\theta | x_i, \tily_i] )$}
\end{algorithmic}
\end{algorithm}

\begin{lemma}[Truthfulness]\label{thm.truthful}
For all $a,b>0$, truthful reporting is a Bayes Nash equilibrium under Algorithm \ref{alg:formal}.
\end{lemma}
\begin{proof}
Recall that conditioned on $x_i,y_i$, the distribution of $X_{-i},y_{-i}$ is independent of $c_i$. Hence, assuming all other players are truthful,  player $i$'s expected payment conditioned on her data $(x_i, y_i)$ and her cost $c_i$, for reporting $\tily_i$ is, $$\E[\pi_i | x_i, y_i, c_i] = \E\left[B_{a,b}(x_i^{\top}\lr_{-i} , x_i^{\top}\E[\theta | x_i, \tily_i ] ) | x_i, y_i \right] 
= B_{a,b}\left( x_i^\top\E[\lr_{-i}| x_i, y_i], x_i^\top\E[\theta | x_i, \tily_i ] \right).$$
The second inequality is due to the linearity of $B_{a,b}$ in its first argument, as well as the linearity of the inner product.
Note that $B_{a,b}$  is uniquely maximized by reporting $\tily_i$ such that $\E[\theta | x_i, \tily_i ]^{\top} x_i = \E[\lr_{-i}| x_i, y_i ]^{\top} x_i$.  Since $\lr$ is an unbiased estimator of $\theta$, then $\E[\lr_{-i}| x_i, y_i ] = \E[\theta | x_i, y_i ]$.  Thus the optimal misreport is $\tily_i$ such that $\E[\theta | x_i, \tily_i ]^{\top} x_i = \E[\theta | x_i, y_i ]^{\top} x_i$, so truthful reporting is a Bayes Nash equilibrium.
\end{proof}

We note that truthfulness is essentially a consequence  of  (1) the fact that $B_{a,b}$ is a strictly proper scoring rule (as it is positive-affine in its first argument and strictly concave in its second argument), and  (2) most importantly,  the fact that $\lr_{-i}$ is an unbiased estimator of $\theta$. Moreover, as in the case of the simple peer prediction setting presented in Appendix~\ref{Brier}, truthfulness persists even if $\lr_{-i}$ in Algorithm~\ref{alg:formal} is replaced by a linear regression estimator constructed over responses restricted to an arbitrary set $S\subseteq [n]\setminus i$.

Truthful reports enable accurate computation of the estimator with high probability, with accuracy parameter $\eta = O(\frac{1}{n})$.
\begin{lemma}[Accuracy]\label{thm.acc}
Under truthful reporting, with probability at least $1-d^{-t^2}$ and when $n \geq C(\frac{t}{\xi})^2 (d+2) \log d$, the accuracy the estimator $\lr$  in Algorithm \ref{alg:formal} is 
$\E\left[ \left\| \lr - \theta \right\|_2^2 \right] \leq \frac{\sigma^2}{(1-\xi) \frac{1}{d+2} n}. $
\end{lemma}

\begin{proof}
Note that
$\E\left[ \left\| \lr - \theta \right\|_2^2 \right] 
= \tr( \Cov(\lr)) \stackrel{\eqref{covbias}}{=} \sigma^2 \tr\left((X^{\top}X)^{-1}\right).$
For i.i.d.~features $x_i$, the spectrum of matrix $X^{\top}X$ can be asymptotically characterized by a theorem of \cite{Ver11} (see Theorem~\ref{thm.converge} in Appendix~\ref{sec:linearreg}), and the lemma follows.
\end{proof}

\paragraph{Remark}{Note that individual rationality and a small budget  can be trivially attained in the absence of privacy costs. To ensure individual rationality of Algorithm \ref{alg:formal}, payments $\pi_i$ must be non-negative, but can be made arbitrarily small.  Thus payments can be scaled down to reduce the analyst's total budget.  For example, setting $a = b(B+2B(B+M)+(B+M)^2-1)$ and $b = \frac{1}{n^2}$ ensures $\pi_i \geq 0$ for all players $i$, and the total required budget is $\frac{1}{n}(2B+4B(B+M)+(B+M)^2) = O(\frac{1}{n})$.}

%% file: privacy.tex
\section{Truthful Regression with Privacy Constraints}\label{s.priv}

As we saw in the previous section, in the absence of privacy concerns, it is possible to devise payments that incentivize truthful reporting.  These payments  compensate players based on how well their report agrees with a response predicted by $\lr$ estimated using other player's reports.

Players whose utilities depend on privacy raise several challenges. Recall that the parameters estimated by the analyst, and the payments made to players, need to satisfy joint differential privacy, 
and hence any estimate of $\theta$ revealed publicly by the analyst or used in a payment must be $\epsilon$-differentially private. Unfortunately, the sensitivity of the linear regression estimator $\lr$ to changes in the input data is, in general, unbounded.  As a result, it is not possible to construct a non-trivial differentially private version of $\lr$ by, e.g., adding noise to its output.

In contrast, differentially private versions of regularized estimators like the ridge regression estimator $\rr$ can be constructed.  Recent techniques have been developed for precisely this purpose, not only for ridge regression but for the broader class of learning through (convex) empirical risk minimization~\citep{CMS11,BST14}. In short, the techniques by \cite{CMS11} and \cite{BST14}  succeed precisely because, for $\gamma>0$, the regularized loss \eqref{regularizedloss} is \emph{strongly convex}. This implies that the sensitivity of $\rr$ is bounded, and a differentially private version of $\rr$ can be constructed by adding noise of appropriate variance or though alternative techniques such as objective perturbation.

The above suggest that a possible approach to constructing a truthful, accurate mechanism in the presence of privacy-conscious players is to modify Algorithm~\ref{alg:formal} by replacing $\lr$ with a ridge regression estimator $\rr$, both with respect to the estimate released globally and to any estimates used in computing payments. Unfortunately, such an approach breaks truthfulness because $\rr$ is a biased estimator. The linear regression estimator $\lr$ ensured that the scoring rule $B_{a,b}$ was maximized precisely when players reported their response variable truthfully.  However, in the presence of an expected bias $\mathtt{b}$, it can easily be seen that the optimal report of player $i$ deviates from truthful reporting by a quantity proportional to $\mathtt{b}^Tx_i$.  

We address this issue for large $n$ using again the concentration result by \cite{Ver11} (see Appendix~\ref{sec:linearreg}). This ensures that for large $n$, the spectrum of $X^\top X$ should grow roughly linearly with $n$, with high probability. By \eqref{covbias}, this implies that as long as $\gamma$ grows more slowly than $n$, the bias term of $\rr$ converges to zero, with high probability. Together, these statements ensure that for an appropriate choice of $\gamma$, we attain approximate truthfulness for large $n$, while also ensuring that the output of our mechanism remains differentially private for all $n$. We formalize this intuition by proving that our mechanism presented in  Section \ref{s.privreg}, based on ridge regression, indeed attains approximate truthfulness for large $n$, while also remaining jointly differentially private.

\subsection{Private Regression Mechanism}\label{s.privreg}

We present our mechanism for private and truthful regression in Algorithm \ref{alg:private}, which is a privatized version of Algorithm \ref{alg:formal}.  We incorporate into our mechanism the Output Perturbation algorithm from \cite{CMS11}, which first computes the ridge regression estimator and then adds noise to the output.  This approach is used to ensure that the mechanism's output satisfies joint differential privacy.

The noise vector $v$ will be drawn according to the following distribution $P_L$, which is a high-dimensional Laplace distribution with parameter $\frac{4B + 2M}{\gamma \eps}$: $ P_L(v) \propto \exp\left(\frac{-\gamma \eps}{4B + 2M} \left\|v \right\|_2\right)$.

\begin{algorithm}[h!]
  \caption{Private Regression Mechanism($\gamma$, $\eps$, $a$, $b$) }
  \label{alg:private}
  \begin{algorithmic}
    \STATE{Solicit reports $X \in \left( \R^d \right)^n$ and $\tily \in \mathbb{R}^n$}
    \STATE{Randomly partition players into two groups, with respective data pairs $(X_0, \tily_0)$ and $(X_1, \tily_1)$}
    \STATE{Analyst computes $\rr = (\gamma I + X^{\top} X)^{-1}X^{\top} \tily$ and $\rr_j = (\gamma I + X^{\top}_j X_j)^{-1}X^{\top}_j \tily_j$ for $j=0,1$}
    \STATE{Independently draw $v, v_0, v_1 \in \mathbb{R}^d$ according to distribution $P_L$ }
    \STATE{Compute estimators $\pr = \rr + v$, $\pr_0 = \rr_0 + v_0$, and $\pr_1 = \rr_1 + v_1$}
    \STATE{Output estimator $\pr$}
    \STATE{Pay each player $i$ in group $j$, $\pi_i = B_{a,b}((\pr_{1-j})^{\top} x_i,\E[\theta|x_i,\tily_i]^{\top} x_i)$ for $j=0,1$}
\end{algorithmic}
\end{algorithm}

Here we state an informal version of our main result.  The formal version of this result is stated in Corollary \ref{cor.formal}, which aggregates and instantiates Theorems \ref{thm.priv}, \ref{thm.equil}, \ref{thm.accuracy}, \ref{thm.privateir}, and \ref{thm.budget}.

\begin{theorem}[Main result (Informal)]\label{thm.informal}
Under Assumptions \ref{a.costs}, \ref{a.indep}, and \ref{a.tail}, there exists ways to set $\gamma$, $\eps$, $a$, and $b$ in Algorithm \ref{alg:private} to ensure that with high probability:
\begin{enumerate}
\item the output of Algorithm \ref{alg:private} is $o(\frac{1}{\sqrt{n}})$-jointly differentially private,
\item it is an $o\left(\frac{1}{n}\right)$-approximate Bayes Nash equilibrium for a $(1-o(1))$-fraction of players to truthfully report their data,
\item the computed estimator $\pr$ is $o(1)$-accurate,
\item it is individually rational for a $(1-o(1))$-fraction of players to  participate in the mechanism, and
\item the required budget from the analyst is $o(1)$.
\end{enumerate}
\end{theorem}

%% file: analysis.tex
\section{Analysis of Algorithm \ref{alg:private}}\label{sec.analysis}

In this section, we flesh out the claims made in Theorem \ref{thm.informal}.  Due to space constraints, all proofs are deferred to Appendix \ref{s.proofs}.

\begin{restatable}[Privacy]{theorem}{privacy}\label{thm.priv}
The mechanism in Algorithm \ref{alg:private} is $2\eps$-jointly differentially private.
\end{restatable}
\paragraph{Proof idea}
We first show that the estimators $\pr$, $\pr_0$, $\pr_1$ together satisfy $2\eps$-differential privacy, by bounding the maximum amount that any player's report can affect the estimators.  We then use the Billboard Lemma (Lemma \ref{lem.billboard} in Appendix \ref{s.billboard}) to show that the estimators, together with the vector of payments, satisfy $2\eps$-joint differential privacy.
\vspace{2mm}

Once we have a privacy guarantee, we can build on this to get truthful participation and hence accuracy. To do so, 
we first show that a symmetric threshold strategy equilibrium exists, in which all agents with cost parameter $c_i$ below some threshold $\tau$ should participate and truthfully report their $y_i$.  We  define $\tau_{\alpha, \beta}$ to be the cost threshold such that (1) with probability $1-\beta$ (with respect to the prior from which costs are drawn), at least a $(1-\alpha)$-fraction of players have cost parameter  $c_i \leq \tau_{\alpha, \beta}$, and (2) conditioned on her own data, each player $i$ believes that with probability $1-\alpha$, any other player $j$ will have cost parameter  $c_j \leq \tau_{\alpha, \beta}$.

\begin{definition}[Threshold $\tau_{\alpha, \beta}$]
Fix a marginal cost distribution $\cC$ on $\{c_i\}$, and let
\[ \tau_{\alpha, \beta}^1 = \inf_{\tau} \left( \Pr_{c \sim \cC} \left[ |\{ i:c_i \leq \tau \}| \geq (1-\alpha)n \right] \geq 1- \beta \right), \]
\[ \tau_{\alpha}^2 = \inf_{\tau} \left(\min_{x_i,y_i} \left( \Pr_{c_j \sim \mathcal{C}|x_i, y_i} [c_j \leq \tau] \right) \geq 1 - \alpha \right). \]
Define $\tau_{\alpha, \beta}$ to be the larger of these thresholds:  $\tau_{\alpha, \beta} = \max \{ \tau_{\alpha, \beta}^1, \tau_{\alpha}^2 \}.$

\end{definition}

We also define the threshold strategy $\sigma_{\tau}$, in which a player reports truthfully if her cost $c_i$ is below $\tau$, and is allowed to misreport arbitrarily if her cost is above $\tau$.

\begin{definition}[Threshold strategy]
Define the threshold strategy $\sigma_{\tau}$ as follows:
\[ \sigma_{\tau}(x_i, y_i, c_i) = \begin{cases} \mbox{Report } \tily_i = y_i & \mbox{ if } c_i \leq \tau, \\ \mbox{Report arbitrary } \tily_i & \mbox{ otherwise. } \end{cases} \]
\end{definition}

We  show that $\sigma_{\tau_{\alpha, \beta}}$ forms a symmetric threshold strategy equilibrium in the Private Regression Mechanism of Algorithm \ref{alg:private}.

\begin{restatable}[Truthfulness]{theorem}{equil}\label{thm.equil}
Fix a participation goal $1-\alpha$, a privacy parameter $\eps$, a desired confidence parameter $\beta$, $\xi \in (0,1)$, and $t \geq 1$.  Then under Assumptions \ref{a.costs} and \ref{a.indep}, with probability $1 - d^{t^2}$ and when $n \geq C(\frac{t}{\xi})^2 (d+2) \log d$, the symmetric threshold strategy $\sigma_{\tau_{\alpha, \beta}}$ is an $\eta$-approximate Bayes-Nash equilibrium in Algorithm \ref{alg:private} for 
\[ \eta = b \left( \frac{\alpha n}{\gamma}(4B+2M) + \frac{\gamma B}{\gamma + (1-\xi) \frac{1}{d+2} n} \right)^2 + \tau_{\alpha, \beta} \eps^2. \]
\end{restatable}
\paragraph*{Proof idea}
There are three primary sources of error which cause the estimator $\pr$ to differ from a player's posterior on $\theta$.  First, ridge regression is a biased estimation technique; second, Algorithm \ref{alg:private} adds noise to preserve privacy; third, players with cost parameter $c_i$ above threshold $\tau_{\alpha, \beta}$ are allowed to misreport their data.  We show how to control the effects of these three sources of error, so that $\pr$ is ``not too far'' from a player's posterior on $\theta$.  Finally, we use strong convexity of the payment rule to show that any player's payment from misreporting is at most $\eta$ greater than 
from truthful reporting.

\begin{restatable}[Accuracy]{theorem}{accuracy}\label{thm.accuracy}
Fix a participation goal $1-\alpha$, a privacy parameter $\eps$, a desired confidence parameter $\beta$, $\xi \in (0,1)$, and $t \geq 1$. Then under the symmetric threshold strategy $\sigma_{\tau_{\alpha, \beta}}$, Algorithm \ref{alg:private} will output an estimator $\pr$ such that with probability at least $1-\beta - d^{-t^2}$, and when $n \geq C(\frac{t}{\xi})^2 (d+2) \log d$,
\[ \E[\| \pr - \theta \|_2^2] = O \left( \left(\frac{\alpha n}{\gamma} + \frac{1}{\gamma \eps}\right)^2 + \left(\frac{\gamma}{n}\right)^2 + \left(\frac{1}{n}\right)^2 + \frac{\alpha n}{\gamma} + \frac{1}{\gamma \eps} \right). \]
\end{restatable}
\paragraph*{Proof idea}
As in 
Theorem \ref{thm.equil}, we control the three sources of error in the estimator $\pr$ --- the bias of ridge regression, the noise added to preserve privacy, and the error due to 
some
players misreporting their data --- this time measuring distance with respect to the expected $L_2$ norm difference.
\vspace{2mm}

\input{analysis2.tex}

%% file: analysis2.tex
We next see that players whose cost parameters are below the threshold $\tau_{\alpha, \beta}$ are incentivized to participate.
\begin{restatable}[Individual Rationality]{theorem}{privateir}\label{thm.privateir}
Under Assumption \ref{a.costs}, the mechanism in Algorithm \ref{alg:private} is individually rational for all players with cost parameters $c_i \leq \tau_{\alpha, \beta}$ as long as,
\[ a \geq \left( \frac{\alpha n}{\gamma}(4B+2M) + \frac{\gamma B}{\gamma + (1-\xi) \frac{1}{d+2} n} + B \right) (b+2bB) + bB^2 + \tau_{\alpha, \beta} \eps^2, \]
regardless of the reports from players with cost coefficients above $\tau_{\alpha, \beta}$.
\end{restatable}
\paragraph*{Proof idea}
A player's utility from participating in the mechanism is her payment minus her privacy cost.  The parameter $a$ in the payment rule is a constant offset that shifts each player's payment.  We lower bound the minimum payment from Algorithm \ref{alg:private} and upper bound the privacy cost of any player with cost coefficient below threshold $\tau_{\alpha, \beta}$.  If $a$ is larger than the difference between these two terms, then any player with cost coefficient below threshold will receive non-negative utility. 
\vspace{2mm}

Finally, we analyze the total cost to the analyst for running the mechanism.
\begin{restatable}[Budget]{theorem}{budget}\label{thm.budget}
The total budget required by the analyst to run Algorithm \ref{alg:private} when players utilize threshold equilibrium strategy $\sigma_{\tau_{\alpha, \beta}}$ is 
\[ \mathcal{B} \leq n \left[a + \left( \frac{\alpha n}{\gamma}(4B+2M) + \frac{\gamma B}{\gamma + (1-\xi) \frac{1}{d+2} n} + B \right) (b+2bB) \right]. \]
\end{restatable}
\paragraph{Proof idea}
The analyst's budget is the sum of all payments made to players in the mechanism.  We upper bound the maximum payment to any player, and the total budget required is at most $n$ times this maximum payment.

\subsection{Formal Statement of Main Result}
In this section, we present our main result, Corollary \ref{cor.formal}, which instantiates Theorems \ref{thm.priv}, \ref{thm.equil}, \ref{thm.accuracy}, \ref{thm.privateir}, and  \ref{thm.budget} with a setting of all parameters to get the bounds promised in Theorem \ref{thm.informal}.  Before stating our main result, we first require the following lemma which asymptotically bounds $\tau_{\alpha, \beta}$ for an arbitrary bounded distribution.  We use this to control the asymptotic behavior of $\tau_{\alpha, \beta}$ under Assumption \ref{a.tail}.

\begin{restatable}{lemma}{tail}\label{lem.tail}
For a cost distribution $\mathcal{C}$ with conditional marginal CDF lower bounded by some function $F$: 
\[ \min_{x_i,y_i} \left( \Pr_{c_j \sim \mathcal{C}|x_i, y_i} [c_j \leq \tau] \right) \geq F(\tau), \]  then
\[ \tau_{\alpha, \beta} \leq \max \{ F^{-1}(1 - \alpha \beta), F^{-1}(1 - \alpha) \}. \]
\end{restatable}

We note that under Assumption \ref{a.tail}, 
Lemma \ref{lem.tail} implies that $ \tau_{\alpha, \beta} \leq \max \{ (\alpha \beta)^{-1/p}, (\alpha)^{-1/p} \}$.
Using this fact, we can state a formal version of our main result.


\begin{restatable}[Main result (Formal)]{corollary}{mainres}\label{cor.formal}
Choose $\delta \in (0, \frac{p}{2+2p})$.  Then under Assumptions \ref{a.costs}, \ref{a.indep}, and \ref{a.tail}, setting $\gamma = n^{1 - \frac{\delta}{2}}$, $\epsilon = n^{-1 + \delta}$, $a = (6B+2M)(1+B)^2 n^{-\frac{3}{2}} + n^{-\frac{3}{2} + \delta}$, and $b = n^{-\frac{3}{2}}$ in Algorithm \ref{alg:private}, and taking $\alpha = n^{-\delta}$, $\beta = n^{-\frac{p}{2} + \delta(1+p)}$, $\xi = 1/2$,  and $t = \sqrt{\frac{n}{4C(d+2)\log d}}$, ensures that with probability $1 - d^{\Theta\left(-n \right)} - n^{-\frac{p}{2} + \delta(1+p)}$:
\begin{enumerate}
\item the output of Algorithm \ref{alg:private} is $O\left(n^{-1 + \delta}\right)$-jointly differentially private, 
\item it is an $O\left(n^{-\frac{3}{2} + \delta }\right)$-approximate Bayes Nash equilibrium for a $1-O\left( n^{-\delta} \right)$ fraction of players to truthfully report their data,
\item the computed estimate $\pr$ is $O \left( n^{-\delta} \right)$-accurate,
\item it is individually rational for a $1-O\left( n^{-\delta} \right)$ fraction of players to  participate in the mechanism, and
\item the required budget from the analyst is $O\left(n^{-\frac{1}{2} + \delta}\right)$.
\end{enumerate}
\end{restatable}

This  follows from instantiating Theorems \ref{thm.priv}, \ref{thm.equil}, \ref{thm.accuracy}, \ref{thm.privateir}, and  \ref{thm.budget} with the specified parameters.  Note that the choice of $\delta$ controls the trade-off between approximation factors for the desired properties.



\begin{paragraph}{Remark}
Note that different settings of parameters can be used to yield a different trade-off between approximation factors in the above result.  For example, if the analyst is willing to supply a higher budget (say constant or increasing with $n$), he could improve on the accuracy guarantee. 
\end{paragraph}

%% file: prelim-appendix.tex
\section{Technical Preliminaries}~\label{sec:prelim-appendix}
\subsection{Peer Prediction and the Brier Scoring Rule}~\label{Brier}
\emph{Peer prediction} \citep{MRZ05} is a useful method of inducing truthful reporting among players that hold data generated by the same statistical model. In short, each player reports her data to an analyst and is paid based on how well her report predicts the report of other players; tying each player's payment to how closely it predicts peer reports is precisely what induces truthfulness.
\cite{GLRS14} illustrate these ideas in the context of privacy-sensitive individuals through the use of the Brier scoring rule \citep{Bri50} as a payment scheme among players holding a random bit. As we make use of the same technique, we review here how the Brier scoring rule can be used for basic peer prediction.

The basic Brier scoring rule was designed for the prediction of a binary event.  Let $I$ be an indicator of the event occurring.  Then the payment for reporting that the event will occur with probability $q$ is,
\[ BasicBrier(I,q) = 2Iq + 2(1-I)(1-q) - q^2 - (1-q)^2. \]
Following \cite{GLRS14}, we define an extension of the basic Brier scoring rule.  For any $p$ and $q$, we define the payment function $B(p,q)$ as follows:
\[ B(p,q) = 1 - 2(p - 2pq + q^2) \]
Note that for the prediction of a binary event, $B(p,q)$ is the expected payment according to $BasicBrier(I,q)$ when the event will occur with probability $p$ and the agent submits prediction probability $q$.  That is, $B(p,q) = \E_{I \sim p}[BasicBrier(I,q)]$. By design, $B(p,q)$ is a \emph{strictly proper} scoring rule, which means it is uniquely maximized by a player truthful reporting her belief $q$ about the probability of the event occurring.

Algorithms \ref{alg:formal} and \ref{alg:private} use payment rule $B_{a,b}(p,q)$, which is a parametrized rescaling of the scoring rule $B(p,q)$, defined as follows:
\[ B_{a,b}(p,q) = a - b \left(p - 2pq + q^2\right). \]
Any positive-affine transformation of a strictly proper scoring rule remains strictly proper \citep{Bik07}.  The rescaled Brier scoring rule satisfies this criterion as $B_{a,b}(p,q) = a' + b' B(p,q)$ where $a' = a - b/2$ and $b' = b/2 > 0$.  Thus $B_{a,b}(p,q)$ is a strictly proper scoring rule, and is uniquely maximized by reporting the true probability $q=p$. 

For concreteness, we now provide an example to demonstrate how the payment rule $B(p,q)$ can be used in peer prediction to truthfully elicit players' beliefs.  Consider a set of $n$ players, each holding a binary variable $b_i\in\{0,1\}$. Assume that each of these variables is generated by independent Bernouli trials with parameter $p$, i.e., $\Pr(b_i=1)=p$, for every $i\in [n].$ We assume here that
$p$ is itself a random variable generated from a known prior over $[0,1]$. Each player reports a bit $\tilde{b}_i\in \{0,1\}$ to the analyst, who wishes to estimate $p$ as $\frac{1}{n}\sum_{i\in [n]}\tilde{b}_i$. The analyst therefore wishes to incentivize truthful reporting of the bits $b_i$, through an appropriate payment scheme.

Let $\E[p\mid b]$ be expected value of $p$ conditioned on observing that a player's bit is $b\in \{0,1\}$.  Put differently, for every player whose bit is $b$, $\E[p\mid b]$ captures her belief about the realization of $p$ after she observes her own bit. Consider the following payment rule. To generate the payment for player $i$, the analyst selects a player $j$ uniformly at random from $[n]\setminus i$ and pays player $i$:
\begin{align}B( \tilde{b}_j, \E[p\mid \tilde{b}_i]  )\label{peer}\end{align}
\begin{lemma} \citep{MRZ05} Under payments \eqref{peer}, truthful reporting is a Bayes-Nash equilibrium.
\end{lemma}
\begin{proof} Observe that for all $q,q'\in [0,1]$, $B(q',q)$ is positive, so payments \eqref{peer} are individually rational.  Moreover, for all $q'\in [0,1]$, $B(q',q)$ is a strictly concave function of $q$ maximized at $q'=q$.  Moreover, $B(q',q)$ is an affine function of $q'$.  If player $i$'s bit is $b_i$ and all other players report their bits truthfully (i.e., $\tilde{b}_j=b_j$ for all $j\neq i$), then player $i$'s expected payment is
$\E\left[B(b_j, \E[p\mid \tilde{b}_i] )\mid b_i\right] = B\left(\E[b_j\mid b_i]  , \E[p\mid \tilde{b}_i]\right) = B\left( \E[p\mid b_i], \E[p \mid \tilde{b}_i]\right).$
Hence, player $i$'s payment is maximized when $\tilde{b}_i = b_i.$
\end{proof}
Informally, the payment scheme \eqref{peer}  induces truthfulness by awarding a player the highest payment if the belief induced on $p$ by her reported bit ``agrees'' with the belief induced by the bit of a random peer. We note that instead of the bit of a peer selected at random, any quantity whose expectation conditioned on $b_i$ would be equal to $\E[p\mid b_i]$ would work as input to the payment rule. For example, using the average value
$\bar{b}_S =\frac{1}{|S|}\sum_{j\in S}\tilde{b}_j $
for any $S\subseteq [n]\setminus i$ as the first argument of $B$ would also induce truthful reporting.

\subsection{Properties of ridge regression}\label{sec:linearreg}
As mentioned in Section \ref{s.linreg}, the ridge regression estimator $\rr$ is biased, while the linear regression estimator $\lr$ is unbiased.  Nevertheless, in practice $\rr$ is preferable to $\lr$ as it can achieve a desirable trade-off between \emph{bias} and \emph{variance}. In particular, consider the square loss error of the estimation $\rr$, namely, $\E[\|\rr-\theta\|_2^2]$. If we condition on the true parameter vector $\theta$ and the features $X$, this can be written as 
\begin{align}\E[\|\rr-\theta\|_2^2 ] = \E[\|\rr - \E[\rr]\|_2^2 ] +\|\E[\rr]- \theta\|_2^2= \tr( \Cov(\rr)) + \|\bias(\rr)\|_2^2\end{align}
where $\Cov(\rr)=\E[(\rr-\E[\rr]) (\rr-\E[\rr])^\top]$ and $\bias(\rr)=\E[\rr]-\theta$ are the covariance and bias, respectively, of estimator $\rr$. Assuming that the responses $y$ follow \eqref{linear}\footnote{i.e., under truthful reporting.}, then conditioned on $X$ and $\theta$, these can be computed in closed form as:
\begin{align}\Cov(\rr)&=\sigma^2(\gamma I+X^\top X)^{-1}X^\top X(\gamma I +X^\top X)^{-1},&  \bias(\rr)&=-\gamma(\gamma I+X^\top X)^{-1}\theta, \label{covbias}\end{align}
where $\sigma^2$ is the variance of the noise variables $z_i$ in \eqref{linear}.
It is easy to see that decreasing $\gamma$ decreases the bias, but may significantly increase the variance. For example in the case where $\rank(X)<d$, the matrix $X^\top X$ is not invertible, and the trace of the covariance tends to infinity as $\gamma$ tends to zero.

Whether  $\tr( \Cov(\rr))$ is large and, therefore, whether regularizing the square loss is necessary, depends on largest eigenvalue (i.e., the \emph{spectral norm}) of $(X^\top X)^{-1}$.  Although this can be infinite for arbitrary $X$, if the $x_i$'s are drawn i.i.d.~we expect that as $n$ increases we will get estimates of lower variance. Indeed,
 by the law of large numbers, we expect that if we sample the features $x_i$ independently from an isotropic distribution, then $\frac{1}{n}(X^\top X)$ should converge to the covariance of this distribution (namely $\Sigma = cI$ for some constant $c$). As such, for large $n$ both the largest and smallest eigenvalues of $X^\top X$ should be of the order of $n$, leading to an estimation of ever decreasing variance even when $\gamma =0$. 
The following theorem, which follows as a corollary of a result by~\cite{Ver11}  (see Appendix~\ref{app:vers}), formalizes this notion, providing bounds on both the largest and smallest eigenvalue of $X^\top X$ and $\gamma I +X^\top X$.
\begin{restatable}{theorem}{converge}\label{thm.converge}
Let $\xi \in (0,1)$, and $t \geq 1$. Let $\|\cdot\|$ denote the spectral norm.  If $\{x_i\}_{i\in [n]}$ are i.i.d.~and sampled uniformly from the unit ball, then with probability at least $1-d^{-t^2}$, when $n \geq C(\frac{t}{\xi})^2 (d+2) \log d$, for some absolute constant $C$, then, 
\[ \left\|X^{\top} X \right\| \leq (1+\xi)\frac{1}{d+2} n \mbox{, and } \left\|( X^{\top} X )^{-1}\right\| \leq \frac{1}{(1-\xi)\frac{1}{d+2} n} \mbox{, and} \]
\[ \left\| \gamma I + X^{\top} X \right\| \leq \gamma + (1+\xi)\frac{1}{d+2} n \mbox{, and } \left\|(\gamma I + X^{\top} X )^{-1}\right\| \leq \frac{1}{\gamma + (1-\xi)\frac{1}{d+2} n}. \]
\end{restatable}
\paragraph{Remark}{A generalization of Theorem \ref{thm.converge} holds for $\{x_i\}_{i\in [n]}$ sampled from any distribution with a covariance $\Sigma$ whose smallest eigenvalue is bounded away from zero (see~\cite{Ver11}). We restrict our attention to the unit ball for simplicity and concreteness. 

\subsection{The Billboard Lemma}\label{s.billboard}

A very useful result regarding jointly differentially private mechanisms that we use in our analysis is the so-called ``billboard-lemma'': 
\begin{lemma}[Billboard Lemma \citep{HHRRW14}]\label{lem.billboard}
Let $\cM: \mathcal{D}^n \rightarrow \cO$ be an $\eps$-differentially private mechanism. Consider a set of $n$ functions $h_i: \cD \times \cO \rightarrow \cR$, for $i\in [n]$. Then, the mechanism $\mathcal{M}' : \mathcal{D}^n \to \cO \times \cR^n$ that computes $r = \mathcal{M}(D)$ and outputs $\mathcal{M}'(D) = (r, h_1(\Pi_2 D, r), \ldots, h_n(\Pi_n D, r) )$,  where $\Pi_i$ is the projection to player $i$'s data, is $\eps$-jointly differentially private.
\end{lemma}
In short, the billboard lemma implies that if we can construct payments such that the payment to player $i$ depends only on her data (e.g. $x_i$, $y_i$) and a universally observable output that is $\epsilon$-differentially private (e.g., $\hatth$), then the resulting mechanism will be $\epsilon$-jointly differentially private.

%

%% file: appendixproofs.tex
\section{Proofs from Section \ref{sec.analysis}}\label{s.proofs}

\subsection{Proof of Theorem \ref{thm.priv} (Privacy)}

We will now prove that the estimator $\pr$ and the vector of payments $\pi$ of the mechanism in Algorithm \ref{alg:private} is $2\eps$-jointly differentially private.  First, we need the following lemma to bound the \emph{sensitivity} of $\pr$, formally defined in Definition \ref{def.sens}, which is the maximum change in the output when a single player misreports her data.  For vector-valued outputs, we measure this change with respect to the $L_2$ norm.

\begin{definition}[Sensitivity]\label{def.sens}
The \emph{sensitivity} of a function $f: \mathcal{D} \to \mathcal{R}$ is the maximum $L_2$ norm of the function's output, when a single player changes her input:
\[ \mbox{Sensitivity of }f = \max_{D, D', \; neighbors} \| f(D) - f(D') \|_2 \]
\end{definition}

The following lemma follows from \cite{CMS11}; a proof is provided for completeness.

\begin{lemma}\label{lem.sensitivity}
The sensitivity of $\rr$ is $\frac{1}{\gamma}(4B + 2M)$.
\end{lemma}
\begin{proof}
Let $(X,y)$ and $(X',y')$ be two arbitrary neighboring databases that differ only in the $i$-th entry.  Let $\rr$ and $(\rr)'$ respectively denote the ridge regression estimators computed on $(X,y)$ and $(X',y')$.  Define $g(\theta)$ to be the change in loss when $\theta$ is used as an estimator for $(X',y')$ and $(X,y)$. 
\begin{align*}
g(\theta) &= \cL(\theta;X',y') - \cL(\theta;X,y) \\
&= \left(\theta^{\top} x_i - y_i \right)^2 - \left(\theta^{\top} x_i' - y_i' \right)^2 
\end{align*}
Lemma 7 of \cite{CMS11} says that if $\cL(\theta;X,y)$ and $\cL(\theta;X',y')$ are both $\Gamma$-strongly convex, then $\left\| \rr - (\rr)' \right\|_2$ is bounded above by $\frac{1}{\Gamma} \cdot \max_{\theta} \left\|\nabla g(\theta)\right\|_2$.  By Lemma \ref{lem.strongconv} (in Appendix \ref{app:strongconv}), both $\cL(\theta;X,y)$ and $\cL(\theta;X',y')$ are $2\gamma$-strongly convex, so $\left\| \rr - (\rr)' \right\|_2 \leq \frac{1}{2 \gamma} \cdot \max_{\theta} \left\|\nabla g(\theta)\right\|_2$.  We now bound $\left\|\nabla g(\theta)\right\|_2$ for an arbitrary $\theta$.
\begin{align*}
\left\|\nabla g(\theta)\right\|_2 &= 2 \left\| (\theta^{\top} x_i - y_i )x_i - (\theta^{\top} x_i' - y_i' )x_i' \right\|_2 \\
&\leq 4 \left| \theta^{\top} x_i - y_i \right| \left\|x_i\right\|_2 \\
&\leq 4 \left( \left| \theta^{\top} x_i\right| + \left| y_i \right| \right)\\
&\leq 4 ( 2B + M )
\end{align*}
Since this bound holds for all $\theta$, it must be the case that $\max_{\theta} \left\|\nabla g(\theta)\right\|_2 \leq 4 ( 2B + M )$ as well.  Then by Lemma 7 of \cite{CMS11}, 
\[ \left\| \rr - (\rr)' \right\|_2 \leq \frac{4}{2 \gamma} ( 2B + M ) = \frac{1}{\gamma} ( 4B + 2M ). \]
Since $(X,y)$ and $(X',y')$ were two arbitrary neighboring databases, this bounds the sensitivity of the computation.  Thus changing the input of one player can change the ridge regression estimator (with respect to the $L_2$ norm) by at most $\frac{1}{\gamma} (4B + 2M)$.
\end{proof}

We now prove that the output of Algorithm \ref{alg:private} satisfies $2\eps$-joint differential privacy.

\privacy*



\begin{proof}
We begin by showing that the estimator $\pr$ output by Algorithm \ref{alg:private} is $\eps$-differentially private.

Let $h$ denote the PDF of $\pr$ output  by Algorithm \ref{alg:private}, and $\nu$ denote the PDF of the noise vector $v$.  Let $(X,y)$ and $(X',y')$ be any two databases that differ only in the $i$-th entry, and let $\rr$ and $(\rr)'$ respectively denote the ridge regression estimators computed on these two databases.  

The output estimator $\pr$ is the sum of the ridge regression estimator $\rr$, and the noise vector $v$; the only randomness in the choice of $\pr$ is the noise vector, because $\rr$ is computed deterministically on the data.  Thus the probability that Algorithm \ref{alg:private} outputs a particular $\pr$ is equal to the probability that the noise vector is exactly the difference between $\pr$ and $\rr$.  Fixing an arbitrary $\pr$, let $\hat{v} = \pr - \rr$ and $\hat{v}' = \pr - (\rr)'$.  Then,
\begin{equation}\label{eq:diffpriv} \frac{h(\pr | (X,y))}{h(\pr | (X',y'))} = \frac{\nu(\hat{v})}{\nu(\hat{v}')} = \exp\left(\frac{-\gamma \eps}{8B+4M}(\left\|\hat{v}\right\|_2 - \left\|\hat{v}'\right\|_2) \right)  = \exp\left(\frac{\gamma \eps}{8B+4M}(\left\|\hat{v}'\right\|_2 - \left\|\hat{v}\right\|_2) \right)\end{equation}

By definition, $\pr = \rr + \hat{v} = (\rr)' + \hat{v}'$.  Rearranging terms gives $\rr - (\rr)' = \hat{v}' - \hat{v}$.  By Lemma \ref{lem.sensitivity} and the triangle inequality,
\[ \left\|\hat{v}'\right\|_2 - \left\|\hat{v}\right\|_2 \leq \left\|\hat{v}'-\hat{v}\right\|_2 = \left\|\rr - (\rr)' \right\|_2 \leq \frac{1}{\gamma} (4B + 2M) \]
Plugging this into Equation \eqref{eq:diffpriv} gives the desired inequality,
\[ \frac{h(\pr | (X,y))}{h(\pr | (X',y'))} \leq \exp\left(\frac{\gamma \eps}{4B+2M}\frac{1}{\gamma} (4B + 2M) \right) = \exp(\eps). \]

Next, we show that the output $(\pr, \pr_0, \pr_1, \{\pi_i\}_{i \in [n]})$ of the mechanism satisfies joint differential privacy using the Billboard Lemma.  The estimators $\pr_0$ and $\pr_1$ are computed in the same way as $\pr$, so $\pr_0$ and $\pr_1$ each satisfy $\eps$-differential privacy.  Since $\pr_0$ and $\pr_1$ are computed on disjoint subsets of the data, then by Theorem 4 of \cite{McS09}, together they satisfy $\eps$-differential privacy.  The estimator a player should use to compute her payments depends only on the partition of players, which is independent of the data because it is chosen uniformly at random.   Thus by the Composition Theorem in \cite{DMNS06}, the estimators $(\pr, \pr_0, \pr_1)$ together satisfy $2\eps$-differential privacy.

Each player's payment $\pi_i$ is a function of only her private information --- her report $(x_i, \tily_i)$ and her group in the partition of players --- and the $2\eps$-differentially private vector of estimators $(\pr, \pr_0, \pr_1)$.  Then by the Billboard Lemma \ref{lem.billboard}, the output $(\pr, \pr_0, \pr_1, \{\pi_i\}_{i \in [n]})$ of Algorithm \ref{alg:private} satisfies $2\eps$-joint differential privacy.
\end{proof}

\subsection{Proof of Theorem \ref{thm.equil} (Truthfulness)}\label{s.equil}

In order to show that $\sigma_{\tau_{\alpha, \beta}}$ is an approximate Bayes-Nash equilibrium, we require the following three lemmas.  Lemma \ref{lem.alphafrac} bounds the expected number of players who will misreport under the strategy profile $\sigma_{\tau_{\alpha, \beta}}$.  Lemma \ref{lem.alpha} bounds the norm of the expected difference of two estimators output by Algorithm \ref{alg:private} run on different datasets, as a function of the number of players whose data differs between the two datasets.  Lemma \ref{prop.noise} bounds the first two moments of the noise vector that is added to preserve privacy.

\begin{lemma}\label{lem.alphafrac}
Under symmetric strategy profile $\sigma_{\tau_{\alpha, \beta}}$, each player expects that at most an $\alpha$-fraction of other players will misreport, given Assumption \ref{a.indep}.
\end{lemma}

\begin{proof}
Let $S_{-i}$ denote the set of players other than $i$ who truthfully report under strategy $\sigma_{\tau_{\alpha, \beta}}$.  From the perspective of player $i$, the cost coefficients of all other players are drawn independently from the posterior marginal distribution $\mathcal{C}|_{x_i, y_i}$.  By the definition of $\tau_{\alpha, \beta}$, player $i$ believes that each other player truthfully reports independently with probability at least $1 - \alpha$.  Thus $\E[ |S_{-i}| \; | x_i, y_i] \geq (1-\alpha)(n-1)$.
\end{proof}

\begin{lemma}\label{lem.alpha}
Let $\rr$ and $(\rr)'$ be the ridge regression estimators on two fixed databases that differ on the input of at most $k$ players.  Then 
\[ \left\| \rr - (\rr)' \right\|_2 \leq \frac{k}{\gamma}(4B+2M) \]
\end{lemma}
\begin{proof}
Since the two databases differ on the reports of at most $k$ players, we can define a sequence of databases $D_0, \ldots, D_k$, that each differ from the previous database in the input of at most one player, and $D_0$ is the input that generated $\rr$, and $D_k$ is the input that generated $(\rr)'$.  Consider running Algorithm \ref{alg:private} on each database $D_j$ in the sequence.  For each $D_j$, let $\rr_j$ be the ridge regression estimator computed on $D_j$.  Note that $\rr_0 = \rr$ and $\rr_k = (\rr)'$.
\begin{align*}
\left\| \rr - (\rr)' \right\|_2 &= \left\| \rr_0-\rr_k \right\|_2 \\
&= \left\| \rr_0 - \rr_1 + \rr_1 - \ldots - \rr_{k-1} + \rr_{k-1} - \rr_k \right\|_2 \\
&\leq \left\| \rr_0 - \rr_1 \right\|_2 + \left\| \rr_1 - \rr_2 \right\|_2 + \ldots + \left\| \rr_{k-1} - \rr_k \right\|_2 \\
&\leq k \cdot \max_j \left\| \rr_j - \rr_{j+1} \right\|_2
\end{align*}
For each $j$, $\rr_j$ and $\rr_{j+1}$ are the ridge regression estimators computed on databases that differ in the data of at most a single player.  That means either the databases are the same, so $\rr_j=\rr_{j+1}$ and their normed difference is $0$, or they differ in the report of exactly one player.  In the latter case, Lemma \ref{lem.sensitivity} bounds $\| \rr_j - \rr_{j+1} \|_2$ above by $\frac{1}{\gamma}(4B+2M)$ for each $j$, including the $j$ which maximizes the normed difference.

Combining this fact with the above inequalities gives,
\[ \left\| \rr - (\rr)' \right\|_2 \leq \frac{k}{\gamma}(4B+2M). \]
\end{proof}

\begin{lemma}\label{prop.noise}
The noise vector $v$ added in Algorithm \ref{alg:private} satisfies: $\E[v] = \vec{0}$ and $\E[\|v\|_2^2] = 2\left(\frac{4B + 2M}{\gamma \epsilon}\right)^2$ and $\E[\|v\|_2] = \frac{4B + 2M}{\gamma \epsilon}$.
\end{lemma}
\begin{proof}
For every $\bar{v} \in \R^d$, there exists $-\bar{v} \in \R^d$ that is drawn with the same probability, because $\| \bar{v} \|_2 = \| -\bar{v} \|_2$.  Thus,
\[ \E[ v] = \int_{\bar{v}} \bar{v} \; \Pr(v = \bar{v}) d\bar{v} = \frac{1}{2} \int_{\bar{v}} (\bar{v} + -\bar{v}) \; \Pr(v = \bar{v}) d\bar{v} = \vec{0}.\]

The distribution of $v$ is a high dimensional Laplacian with parameter $\frac{4B + 2M}{\gamma \epsilon}$ and mean zero.  It follows immediately that $\E[\|v\|_2^2] = 2\left(\frac{4B + 2M}{\gamma \epsilon}\right)^2$ and $\E[\|v\|_2] = \frac{4B + 2M}{\gamma \epsilon}$.
\end{proof}

We now prove that symmetric threshold strategy $\sigma_{\tau_{\alpha, \beta}}$ is an approximate Bayes-Nash equilibrium in Algorithm \ref{alg:private}.


\equil*

\begin{proof}
Suppose all players other than $i$ are following strategy $\sigma_{\tau_{\alpha, \beta}}$.  Let player $i$ be in group $1-j$, so she is paid according to the estimator computed on the data of group $j$.  Let $\pr_j$ be the estimator output by Algorithm \ref{alg:private} on the reported data of group $j$ under this strategy, and let $(\rr_j)'$ be the ridge regression estimator computed within Algorithm \ref{alg:private} when all players in group $j$ follow strategy $\sigma_{\tau_{\alpha, \beta}}$.  Let $\rr_j$ be the ridge regression estimator that would have been computed within Algorithm \ref{alg:private} if all players in group $j$ had reported truthfully.  For ease of notation, we will suppress the subscripts on the estimators for the remainder of the proof.

We will show that $\sigma_{\tau_{\alpha, \beta}}$ is an approximate Bayes-Nash equilibrium by bounding player $i$'s incentive to deviate.  We assume that $c_i \leq \tau_{\alpha, \beta}$ (otherwise there is nothing to show because player $i$ would be allowed to submit an arbitrary report under $\sigma_{\tau_{\alpha, \beta}}$).  We first compute the maximum amount that player $i$ can increase her payment by misreporting to Algorithm \ref{alg:private}.  Consider the expected payment to player $i$ from a fixed (deterministic) misreport, $\tily_i = y_i + \delta$.  
\begin{align*}
\E&[B_{a,b}((\pr)^{\top} x_i, \E[\theta|x_i,\tily_i]^{\top} x_i) | x_i, y_i] - \E[B_{a,b}((\pr)^{\top} x_i, \E[\theta|x_i,y_i]^{\top} x_i) | x_i, y_i] \\
&= B_{a,b}(\E[\pr | x_i,y_i]^{\top} x_i, \E[\theta|x_i,\tily_i]^{\top} x_i) - B_{a,b}(\E[\pr | x_i, y_i]^{\top} x_i, \E[\theta|x_i,y_i]^{\top} x_i)
\end{align*}

The rule $B_{a,b}$ is a proper scoring rule, so it is uniquely maximized when its two arguments are equal.  Thus any misreport of player $i$ cannot yield payment greater than $B_{a,b}(\E[\pr | x_i,y_i]^{\top} x_i, \E[\pr | x_i,y_i]^{\top} x_i)$, so the expression of interest is bounded above by the following.
\begin{align*}
B_{a,b}&(\E[\pr | x_i,y_i]^{\top} x_i, \E[\pr | x_i,y_i]^{\top} x_i) - B_{a,b}(\E[\pr | x_i, y_i]^{\top} x_i, \E[\theta|x_i,y_i]^{\top} x_i) \\
&= a-b\left(\E[\pr | x_i,y_i]^{\top} x_i - 2(\E[\pr | x_i,y_i]^{\top} x_i)^2 + (\E[\pr | x_i,y_i]^{\top} x_i)^2\right)\\
& \; \; \;  -a +b\left(\E[\pr | x_i,y_i]^{\top} x_i - 2(\E[\pr | x_i,y_i]^{\top} x_i)(\E[\theta|x_i,y_i]^{\top} x_i) + (\E[\theta|x_i,y_i]^{\top} x_i)^2 \right) \\
&= b \left( (\E[\pr | x_i,y_i]^{\top} x_i)^2 - 2(\E[\pr | x_i,y_i]^{\top} x_i)(\E[\theta|x_i,y_i]^{\top} x_i) + (\E[\theta|x_i,y_i]^{\top} x_i)^2 \right) \\
&= b \left( \E[\pr | x_i,y_i]^{\top} x_i - \E[\theta|x_i,y_i]^{\top} x_i \right)^2 \\
&= b \left( \E[\pr - \theta | x_i,y_i]^{\top} x_i \right)^2 \\
&\leq b ( \| \E[\pr - \theta | x_i,y_i] \|_2^2 \|x_i\|_2^2 ) \\
&\leq b \| \E[\pr - \theta | x_i,y_i] \|_2^2
\end{align*}

We continue by bounding the term $\| \E[\pr - \theta | x_i,y_i] \|_2$.
\begin{align*}
\| \E[\pr - \theta | x_i,y_i] \|_2 &= \| \E[\pr - \rr + \rr - \theta | x_i,y_i] \|_2 \\
&= \| \E[(\rr)' + v - \rr + \rr - \theta | x_i,y_i] \|_2 \\
&= \| \E[v | x_i,y_i] + \E[(\rr)' - \rr | x_i,y_i] + \E[\rr - \theta | x_i,y_i] \|_2 \\
&\leq \| \E[v | x_i,y_i] \|_2 + \| \E[(\rr)' - \rr | x_i,y_i] \|_2 + \| \E[\rr - \theta | x_i,y_i] \|_2 
\end{align*}

We again bound each term separately.  In the first term, the noise vector is drawn independently of the data, so $\E[v | x_i,y_i]= \E[v]$, which equals $\vec{0}$ by Lemma \ref{prop.noise}.  Thus $\| \E[v | x_i,y_i] \|_2 = 0$.

Jensen's inequality bounds the second term above by $\E[ \| (\rr)' - \rr \|_2 | x_i,y_i]$. The random variables $(\rr)'$ and $\rr$ are the ridge regression estimators of two (random) databases that differ only on the data of players who misreported under threshold strategy $\sigma_{\tau_{\alpha, \beta}}$.  By Lemma \ref{lem.alphafrac}, player $i$ believes that at most $\alpha n$ players will misreport their $\tily_j$,\footnote{Lemma \ref{lem.alphafrac} promises that at most $\alpha(n-1)$ players will misreport.  We use the weaker bound of $\alpha n$ for simplicity.}  so for all pairs of databases over which the expectation is taken, $(\rr)'$ and $\rr$ differ in the input of at most $\alpha n$ players.  By Lemma \ref{lem.alpha}, their normed difference is bounded above by $\frac{\alpha n}{\gamma}(4B+2M)$.  Since this bound applied to every term over which the expectation is taken, it also bounds the expectation. 

For the third term, $\E[\rr - \theta | x_i,y_i] = \bias(\rr | x_i, y_i).$  Recall that $\rr$ is actually $\rr_j$, which is computed independently of player $i$'s data, but is still correlated with $(x_i, y_i)$ through the common parameter $\theta$.  However, conditioned on the true $\theta$, the bias of $\rr$ is independent of player $i$'s data.  That is, $\bias(\rr | x_i, y_i, \theta) = \bias(\rr | \theta)$.   We now expand the third term using nested expectations.
\begin{align*}
\E_{X, z, \theta}\left[\rr - \theta | x_i,y_i\right] &= \E_{\theta}\left[ \E_{X,z}[\rr - \theta | x_i,y_i, \theta] \right] \\
&= \E_{\theta}\left[ \bias(\rr | x_i, y_i, \theta) \right] \\
&= \E_{\theta}\left[ \bias(\rr | \theta) \right] \\
&= \bias(\rr) \\
&= -\gamma (\gamma I + X^{\top}X)^{-1}\theta
\end{align*}
Then by Theorem \ref{thm.converge}, when $n \geq C(\frac{t}{\xi})^2 (d+2) \log d$, the following holds with probability at least $1-d^{-t^2}$.
\begin{align*}
\| \E[\rr - \theta | x_i,y_i] \|_2 &= \| -\gamma (\gamma I + X^{\top}X)^{-1}\theta \|_2 \\
&\leq \gamma \| (\gamma I + X^{\top}X)^{-1} \|_2 \| \theta \|_2 \\
&\leq \gamma \left( \frac{1}{\gamma + (1-\xi) \frac{1}{d+2} n} \right) B \\
&= \frac{\gamma B}{\gamma + (1-\xi) \frac{1}{d+2} n}
\end{align*}

We will assume the above is true for the remainder of the proof, which will be the case except with probability at most $d^{-t^2}$.  Thus with probability at least $1- d^{-t^2}$, and when $n$ is sufficiently large, the increase in payment from misreporting is bounded above by
\[ b \| \E[\pr - \theta | x_i,y_i] \|_2^2 \leq b \left( \frac{\alpha n}{\gamma}(4B+2M) + \frac{\gamma B}{\gamma + (1-\xi) \frac{1}{d+2} n} \right)^2. \]

In addition to an increased payment, a player may also experience decreased privacy costs from misreporting.  By Assumption \ref{a.costs}, this decrease in privacy costs is bounded above by $c_i \eps^2$.  We have assumed $c_i \leq \tau_{\alpha, \beta}$ (otherwise player $i$ is allowed to misreport arbitrarily under $\sigma_{\tau_{\alpha, \beta}}$, and there is nothing to show).  Then the decrease in privacy costs for player $i$ is bounded above by $\tau_{\alpha, \beta} \eps^2$.

Therefore player $i$'s total incentive to deviate is bounded above by $\eta$, and the symmetric threshold strategy $\sigma_{\tau_{\alpha, \beta}}$ forms an $\eta$-approximate Bayes Nash equilibrium for 
\[ \eta = b \left( \frac{\alpha n}{\gamma}(4B+2M) + \frac{\gamma B}{\gamma + (1-\xi) \frac{1}{d+2} n} \right)^2 + \tau_{\alpha, \beta} \eps^2. \]
\end{proof}

%

\subsection{Proof of Theorem \ref{thm.accuracy} (Accuracy)}

In this section, we prove that the estimator $\pr$ output by Algorithm \ref{alg:private} has high accuracy.   We first require the following lemma, which uses the concentration inequalities of Theorem \ref{thm.converge} to give high probability bounds on the distance from the ridge regression estimator to the true parameter $\theta$.

\begin{lemma}\label{lem.ridge}
Let $\rr$ be the ridge regression estimator computed on a given database $(X,y)$.  Then with probability at least $1-d^{-t^2}$, as long as $n \geq C(\frac{t}{\xi})^2 (d+2) \log d$
\[ \E[ \| \rr - \theta \|_2^2]  \leq \left( \frac{\gamma B}{\gamma + (1-\xi)\frac{1}{d+2}n} \right)^2 + \sigma^4 \left( \frac{ (1+\xi) \frac{1}{d+2} n}{(\gamma + (1-\xi)\frac{1}{d+2}n)^2}\right)^2 \]
and 
\[ \E[ \| \rr - \theta \|_2]  \leq \frac{\gamma B + Mn}{\gamma + (1-\xi)\frac{1}{d+2}n}. \]
\end{lemma}
\begin{proof}
Recall from Section \ref{sec:linearreg} that,
\[ \E[ \| \rr - \theta \|_2^2] = \| \bias(\rr) \|_2^2 + \tr(\Cov(\rr) ), \] 
and,
\begin{align*} \E[ \| \rr - \theta \|_2] &= \E[ \| \rr - \E[\rr] + \E[\rr] - \theta \|_2 ] \\
&\leq \E[ \| \rr - \E[\rr] \|_2] + \E[ \| \E[\rr] - \theta \|_2 ] \\
&= \E[ \| \rr - \E[\rr] \|_2] + \E[ \| \bias(\rr) \|_2 ] 
\end{align*}

We now expand the remaining terms: $\| \bias(\rr) \|_2$ and $\tr( \Cov(\rr) )$ and $\E[ \| \rr - \E[ \rr ] \|_2 ] $.   For the remainder of the proof, we will assume the concentration inequalities in Theorem \ref{thm.converge} hold, which will be the case, except with probability at most $d^{-t^2}$, as long as $n \geq C(\frac{t}{\xi})^2 (d+2) \log d$.

\begin{align*}
\| \bias(\rr) \|_2 &= \| -\gamma(\gamma I + X^{\top}X)^{-1}\theta \|_2 \\
&\leq \gamma \| \theta \|_2 \| (\gamma I + X^{\top}X)^{-1} \|_2 \\
&\leq \gamma B \| (\gamma I + X^{\top}X)^{-1} \|_2 \\
&\leq \frac{\gamma B}{\gamma + (1-\xi)\frac{1}{d+2} n}
\end{align*}

\begin{align*}
\tr(\Cov(\rr)) &= \| \Cov(\rr) \|_2^2 \\
&= \| \sigma^2 (\gamma I + X^{\top} X)^{-1} X^{\top} X (\gamma I + X^{\top} X)^{-1} \|_2^2 \\
&\leq \sigma^4 \| (\gamma I + X^{\top} X)^{-1} \|_2^2 \| X^{\top} X \|_2^2 \|(\gamma I + X^{\top} X)^{-1} \|_2^2 \\
&\leq \sigma^4 \left(\frac{1}{\gamma + (1-\xi)\frac{1}{d+2} n}\right)^2 \left((1+\xi)\frac{1}{d+2} n \right)^2 \left(\frac{1}{\gamma + (1-\xi)\frac{1}{d+2} n}\right)^2 \\
&\leq \sigma^4 \left( \frac{(1+\xi)\frac{1}{d+2} n}{\left(\gamma + (1-\xi)\frac{1}{d+2} n \right)^2} \right)^2
\end{align*}

\begin{align*} \E[ \| \rr - \E[ \rr ] \|_2 ] &= \E[ \| \rr - (\theta + \bias(\rr)) \|_2 ] \\
&= \E[ \| (\gamma I + X^{\top} X)^{-1} X^{\top}y - \theta + (\gamma I + X^{\top}X)^{-1}\gamma I \theta \|_2 ] \\
&= \E[ \| (\gamma I + X^{\top} X)^{-1} X^{\top}(X \theta + z) - \theta + (\gamma I + X^{\top}X)^{-1}\gamma I \theta \|_2 ] \\
&= \E[ \| (\gamma I + X^{\top} X)^{-1} (X^{\top}X + \gamma I)\theta - \theta + (\gamma I + X^{\top}X)^{-1}X^{\top}z \|_2 ] \\
&= \E[ \| \theta - \theta + (\gamma I + X^{\top}X)^{-1}X^{\top}z \|_2 ] \\
&= \E[ \| (\gamma I + X^{\top}X)^{-1}X^{\top}z \|_2 ] \\
&\leq \E[ \|(\gamma I + X^{\top}X)^{-1}\|_2 \| X^{\top} z \|_2 ] \\
&\leq \E[ \|(\gamma I + X^{\top}X)^{-1}\|_2 Mn ]\\
&\leq \frac{Mn}{\gamma + (1-\xi)\frac{1}{d+2} n}
\end{align*}

Using these bounds, we see:
\[ \E[ \| \rr - \theta \|_2^2]  \leq \left( \frac{\gamma B}{\gamma + (1-\xi)\frac{1}{d+2}n} \right)^2 + \sigma^4 \left( \frac{ (1+\xi) \frac{1}{d+2} n}{(\gamma + (1-\xi)\frac{1}{d+2}n)^2}\right)^2 \]
and
\begin{align*} \E[ \| \rr - \theta \|_2]  &\leq \frac{\gamma B}{\gamma + (1-\xi)\frac{1}{d+2}n}  + \frac{Mn}{\gamma + (1-\xi)\frac{1}{d+2} n} \\
&= \frac{\gamma B + Mn}{\gamma + (1-\xi)\frac{1}{d+2}n}
\end{align*}
\end{proof}


We now prove the accuracy guarantee for the estimator $\pr$ output by Algorithm \ref{alg:private}.

\accuracy*

\begin{proof}
Let the data held by players be $(X,y)$, and let $\tily = y + \vec{\delta}$ be the reports of players under the threshold strategy $\sigma_{\tau_{\alpha, \beta}}$. As in Theorem \ref{thm.equil}, let $\pr$ be the estimator output by Algorithm \ref{alg:private} on the reported data under this strategy, and let $(\rr)'$ be the ridge regression estimator computed Algorithm \ref{alg:private} when all players follow strategy $\sigma_{\tau_{\alpha, \beta}}$.  Let $\rr$ be the ridge regression estimator that would have been computed within Algorithm \ref{alg:private} if all players had reported truthfully. Recall that $v$ is the noise vector added in Algorithm \ref{alg:private}.

\begin{align*}
\E[\| \pr - \theta \|_2^2] &= \E[\| \pr - \rr + \rr - \theta \|_2^2] \\
&= \E\left[\| \pr - \rr \|_2^2 + \|\rr - \theta \|_2^2 + 2\left\langle \pr - \rr, \rr - \theta \right\rangle \right] \\
&\leq \E[\| \pr - \rr \|_2^2] + \E[\|\rr - \theta \|_2^2] + 2\E[ \| \pr - \rr \|_2 \| \rr - \theta \|_2] 
\end{align*}

We start by bounding the first term.  Recall that the estimator $\pr$ is equal to the ridge regression estimator on the \emph{reported} data, plus the noise vector $v$ added by Algorithm \ref{alg:private}.

\begin{align*}
\E[\| \pr - \rr \|_2^2] &= \E[\| (\rr)' + v - \rr \|_2^2] \\
&= \E[ \| (\rr)' - \rr \|_2^2 ] + \E[ \| v \|_2^2 ] + 2 \E[ \langle (\rr)' - \rr, v \rangle ] \\
&= \E[ \| (\rr)' - \rr \|_2^2 ] + \E[ \| v \|_2^2 ] + 2 \langle \E [(\rr)' - \rr], \E[ v ] \rangle \\ 
&= \E[ \| (\rr)' - \rr \|_2^2 ] + 2 \left(\frac{4B+2M}{\gamma \eps}\right)^2 \mbox{ (by Lemma \ref{prop.noise})} 
\end{align*}

The estimators $(\rr)'$ and $\rr$ are the ridge regression estimators of two (random) databases that differ only on the data of players who misreported under threshold strategy $\sigma_{\tau_{\alpha, \beta}}$.  The definition of $\tau_{\alpha, \beta}$ ensures us that with probability $1-\beta$, at most $\alpha n$ players will misreport their $\tily_j$.  For the remainder of the proof, we will assume that at most $\alpha n$ players misreported to the mechanism, which will be the case except with probability $\beta$.

Thus for all pairs of databases over which the expectation is taken, $(\rr)'$ and $\rr$ differ in the input of at most $\alpha n$ players, and by Lemma \ref{lem.alpha}, their normed difference is bounded above by $\left(\frac{\alpha n}{\gamma}(4B+2M) \right)^2$.  Since this bound applies to every term over which the expectation is taken, it also bounds the expectation.

Thus the first term satisfies the following bound:
\[ \E[\| \pr - \theta \|_2^2] \leq \left(\frac{\alpha n}{\gamma}(4B+2M) \right)^2 + 2 \left( \frac{4B+2M}{\gamma \eps}\right)^2. \]

By Lemma \ref{lem.ridge}, with probability at least $1- d^{-t^2}$, when $n \geq C(\frac{t}{\xi})^2 (d+2) \log d$, the second term is bounded above by 
\[ \E[ \| \rr - \theta \|_2^2]  \leq \left( \frac{\gamma B}{\gamma + (1-\xi)\frac{1}{d+2}n} \right)^2 + \sigma^4 \left( \frac{ (1+\xi) \frac{1}{d+2} n}{(\gamma + (1-\xi)\frac{1}{d+2}n)^2}\right)^2. \]
We will also assume for the remainder of the proof that the above bound holds, which will be the case except with probability at most $ d^{-t^2}$.

We now bound the third term.
\begin{align*} 2&\E[ \| \pr - \rr \|_2 \| \rr - \theta \|_2] = 2\E[ \| (\rr)' + v - \rr \|_2 \| \rr - \theta \|_2] \\
&\leq 2\E[ \left(\| (\rr)' - \rr \|_2 + \| v \|_2 \right) \| \rr - \theta \|_2] \\
&= 2\E[ \| (\rr)' - \rr \|_2 \| \rr - \theta \|_2 ] + 2\E[ \| v \|_2 \| \rr - \theta \|_2 ]\\
&= 2\E[ \| (\rr)' - \rr \|_2 \| \rr - \theta \|_2 ] + 2\E[ \| v \|_2] \E[ \| \rr - \theta \|_2 ] \mbox{ (by independence)} \\ 
&= 2\E[ \| (\rr)' - \rr \|_2 \| \rr - \theta \|_2 ] + 2 \left(\frac{4B+2M}{\gamma \eps}\right) \E[ \| \rr - \theta \|_2 ] \mbox{ (by Lemma \ref{prop.noise})} 
\end{align*}

We have assumed at most $\alpha n$ players misreported (which will occur with probability at least $1-\beta$), so for all pairs of databases over which the expectation in the first term is taken, Lemma \ref{lem.alpha} bounds $\| (\rr)' - \rr \|$ above by $\frac{\alpha n}{\gamma}(4B+2M)$.  Thus we continue bonding the third term:
\begin{align*}
2\E[ \| (\rr)' - \rr \|_2 & \| \rr - \theta \|_2 ] + 2 \left(\frac{4B+2M}{\gamma \eps}\right) \E[ \| \rr - \theta \|_2 ]  \\
&\leq 2 \E[ \left(\frac{\alpha n}{\gamma}(4B+2M) \right) \| \rr - \theta \|_2 ] + 2 \frac{4B+2M}{\gamma \eps} \E[ \| \rr - \theta \|_2 ] \mbox{ (by Lemma \ref{lem.alpha})} \\
&= 2 \left(\frac{\alpha n}{\gamma}(4B+2M) \right) \E[ \| \rr - \theta \|_2 ] + 2 \frac{4B+2M}{\gamma \eps} \E[ \| \rr - \theta \|_2 ] \\
&= 2 \left( \frac{\alpha n}{\gamma}(4B+2M) + \frac{4B+2M}{\gamma \eps} \right) \E[ \| \rr - \theta \|_2 ] \\
&\leq 2 \left( \frac{\alpha n}{\gamma}(4B+2M) + \frac{4B+2M}{\gamma \eps} \right) \frac{\gamma B + Mn}{\gamma + (1-\xi)\frac{1}{d+2}n} \mbox{ (by Lemma \ref{lem.ridge})} 
\end{align*}

We can now plug these terms back in to get our final accuracy bound.  Taking a union bound over the two failure probabilities, with probability at least $1-\beta - d^{-t^2}$, when $n \geq C(\frac{t}{\xi})^2 (d+2) \log d$:
\begin{align*} \E[\| \pr - \theta \|_2^2] &\leq \left( \frac{\alpha n}{\gamma} (4B+2M) \right)^2 + 2\left( \frac{4B+2M}{\gamma \eps} \right)^2 + \left( \frac{\gamma B}{\gamma + (1-\xi)\frac{1}{d+2}n} \right)^2 \\
& \; \; + \sigma^4 \left( \frac{ (1+\xi) \frac{1}{d+2} n}{(\gamma + (1-\xi)\frac{1}{d+2}n)^2}\right)^2 + 2 \left( \frac{\alpha n}{\gamma} (4B+2M) + \frac{4B+2M}{\gamma \eps} \right) \frac{\gamma B + Mn}{\gamma + (1-\xi)\frac{1}{d+2}n} 
\end{align*}
\end{proof}

%
%

\subsection{Proof of Theorems \ref{thm.privateir} and \ref{thm.budget} (Individual Rationality and Budget)}

In this section we first characterize the conditions needed for individual rationality, and then compute the total budget required from the analyst to run the Private Regression Mechanism in Algorithm \ref{alg:private}.  Note that if we do not require individual rationality, it is easy to achieve a small budget: we can scale down payments as in the non-private mechanism from Section \ref{s.regression}.  However, once players have privacy concerns, they will no longer accept an arbitrarily small positive payment; each player must be paid enough to compensate for her privacy loss.  In order to incentivize players to participate in the mechanism, the analyst will have to ensure that players receive non-negative utility from participation.

We first show that Algorithm \ref{alg:private} is individually rational for players with privacy costs below threshold.  Note that because we allow cost parameters to be unbounded, it is not possible in general to ensure individual rationality for all players while maintaining a finite budget.


\privateir*

\begin{proof}
Let player $i$ have privacy cost parameter $c_i \leq \tau_{\alpha, \beta}$, and consider player $i$'s utility from participating in the mechanism.  Let player $i$ be in group $1-j$, so she is paid according to the estimator computed on the data of group $j$.  Let $\pr_j$ be the estimator output by Algorithm \ref{alg:private} on the reported data of group $j$ under this strategy, and let $(\rr_j)'$ be the ridge regression estimator computed within Algorithm \ref{alg:private} when all players in group $j$ follow strategy $\sigma_{\tau_{\alpha, \beta}}$.  Let $\rr_j$ be the ridge regression estimator that would have been computed within Algorithm \ref{alg:private} if all players in group $j$ had reported truthfully.  For ease of notation, we will suppress the subscripts on the estimators for the remainder of the proof.
\begin{align*}
\E[u_i(x_i,y_i, \tily_i)] &= \E[B_{a,b}((\pr)^\top x_i, \E[ \theta | x_i, \tily_i]^{\top} x_i ) | x_i, y_i] - \E[f_i(c_i, \eps)] \\
&\geq \E[B_{a,b}((\pr)^\top x_i, \E[ \theta | x_i, \tily_i]^{\top} x_i ) | x_i, y_i] -  \tau_{\alpha, \beta} \eps^2 \mbox{ (by Assump. \ref{a.costs})} \\
&= B_{a,b}(\E[\pr |x_i, y_i]^\top x_i, \E[ \theta | x_i, \tily_i]^{\top} x_i ) -  \tau_{\alpha, \beta} \eps^2
\end{align*}

We proceed by bounding the inputs to the payment rule, and thus lower-bounding the payment player $i$ receives.  The second input satisfies the following bound.
\[ \E[ \theta | x_i, \tily_i]^{\top} x_i \leq \| \E[ \theta | x_i, \tily_i] \|_2 \| x_i \|_2 \leq B \]
We can also bound the first input to the payment rule as follows.
\begin{align*}
\E[\pr |x_i, y_i]^\top x_i &= \E[(\rr)' |x_i, y_i]^\top x_i + \E[v |x_i, y_i]^\top x_i \\
&= \E[(\rr)' |x_i, y_i]^\top x_i \\
&\leq \| \E[(\rr)' |x_i, y_i] \|_2 \| x_i \|_2 \\
&\leq \| \E[(\rr)' - \rr |x_i, y_i] \|_2 + \| \E[\rr - \theta |x_i, y_i] \|_2 + \| \E[ \theta |x_i, y_i] \|_2 \\
&\leq \frac{\alpha n}{\gamma}(4B+2M) + \frac{\gamma B}{\gamma + (1-\xi) \frac{1}{d+2} n} + B \mbox{ (by Lemma \ref{lem.alpha} and Theorem \ref{thm.converge})}
\end{align*}



Recall that our Brier-based payment rule is $B_{a,b}(p,q) = a - b \left(p - 2pq + q^2\right)$, which is bounded below by $a - b |p| - 2b |p|\; |q| - b |q|^2 = a - |p|(b+2b|q|) - b|q|^2$.  Using the bounds we just computed on the inputs to player $i$'s payment rule, her payment is at least 
\[ \pi_i \geq a - \left( \frac{\alpha n}{\gamma}(4B+2M) + \frac{\gamma B}{\gamma + (1-\xi) \frac{1}{d+2} n} + B \right) (b+2bB) - bB^2. \]
Thus her expected utility from participating in the mechanism is at least 
\[ \E[u_i(x_i,y_i, \tily_i)] \geq a - \left( \frac{\alpha n}{\gamma}(4B+2M) + \frac{\gamma B}{\gamma + (1-\xi) \frac{1}{d+2} n} + B \right) (b+2bB) - bB^2 - \tau_{\alpha, \beta} \eps^2. \]
Player $i$ will be ensured non-negative utility as long as,
\[ a \geq \left( \frac{\alpha n}{\gamma}(4B+2M) + \frac{\gamma B}{\gamma + (1-\xi) \frac{1}{d+2} n} + B \right) (b+2bB) + bB^2 + \tau_{\alpha, \beta} \eps^2. \]
\end{proof}

The next theorem characterizes the total budget required by the analyst to run Algorithm \ref{alg:private}.

\budget*

\begin{proof}
The total budget is the sum of payments to all players.
\begin{align*}
\mathcal{B} = \sum_{i=1}^n \E[\pi_i] &= \sum_{i=1}^n  \E[B_{a,b}((\pr)^\top x_i, \E[ \theta | x_i, \tily_i]^{\top} x_i ) | x_i, y_i] \\
&= \sum_{i=1}^n B_{a,b}(\E[\pr |x_i, y_i]^\top x_i, \E[ \theta | x_i, \tily_i]^{\top} x_i ) 
\end{align*}
Recall that our Brier-based payment rule is $B_{a,b}(p,q) = a - b \left(p - 2pq + q^2\right)$, which is bounded above by $a + b |p| + 2b |p|\; |q| = a + |p|(b+2b|q|)$.  Using the bounds computed in the proof of Theorem \ref{thm.privateir},  each player $i$ receives payment at most,
\[ \pi_i \geq a + \left( \frac{\alpha n}{\gamma}(4B+2M) + \frac{\gamma B}{\gamma + (1-\xi) \frac{1}{d+2} n} + B \right) (b+2bB). \]
Thus the total budget is at most:
\[ \mathcal{B} = \sum_{i=1}^n \E[\pi_i] \leq n \left(a + \left( \frac{\alpha n}{\gamma}(4B+2M) + \frac{\gamma B}{\gamma + (1-\xi) \frac{1}{d+2} n} + B \right) (b+2bB) \right). \]
\end{proof}

%
%


\subsection{Proof of Lemma \ref{lem.tail} (Bound on threshold $\tau_{\alpha, \beta}$)}

\tail*

\begin{proof}
We first bound $\tau^1_{\alpha, \beta}$.
\begin{align*} \tau^1_{\alpha, \beta} &= \inf_{\tau} \left( \Pr_{c \sim \cC} \left[ |\{ i:c_i \leq \tau \}| \geq (1-\alpha)n \right] \geq 1- \beta \right) \\
&= \inf_{\tau} \left( \Pr_{c \sim \cC} \left[ \left|\{ i:c_i \geq \tau \}\right| \leq \alpha n \right] \geq 1- \beta \right) \\
&= \inf_{\tau} \left( 1 - \Pr_{c \sim \cC} \left[ \left|\{ i:c_i \geq \tau \}\right| \geq \alpha n \right] \geq 1- \beta \right) \\
&= \inf_{\tau} \left( \Pr_{c \sim \cC} \left[ \left|\{ i:c_i \geq \tau \}\right| \geq \alpha n \right] \leq \beta \right) 
\end{align*}

We continue by upper bounding the inner term of the expression.
\begin{align*} 
\Pr_{c \sim \cC} \left[ |\{ i:c_i \geq \tau \}| \geq \alpha n \right] &\leq \frac{\E[ |\{ i:c_i \geq \tau \}|}{\alpha n} \mbox{ (by Markov's inequality)} \\
&= \frac{n \; Pr[c_i \geq \tau ]}{\alpha n} \mbox{ (by independence of costs)} \\
&= \frac{Pr[c_i \geq \tau ]}{\alpha}
\end{align*}

From this bound, if $\frac{Pr[c_i \geq \tau ]}{\alpha} \leq \beta$, then also $\Pr_{c \sim \cC} \left[ |\{ i:c_i \geq \tau \}| \geq \alpha n \right] \leq \beta$.  Thus,
\[ \inf_{\tau} \left( \Pr_{c \sim \cC} \left[ |\{ i:c_i \geq \tau \}| \geq \alpha n \right] \leq \beta \right) \leq \inf_{\tau} \left( \frac{Pr[c_i \geq \tau ]}{\alpha} \leq \beta \right), \]
since the infimum in the first expression is taken over a superset of the feasible region of the latter expression.  Then,
\begin{align*}
\tau^1_{\alpha, \beta} &\leq \inf_{\tau} \left( \frac{Pr[c_i \geq \tau ]}{\alpha} \leq \beta \right) \\
&= \inf_{\tau} \left( Pr[c_i \geq \tau ] \leq \alpha \beta \right) \\
&= \inf_{\tau} \left( 1 - Pr[c_i \leq \tau ] \leq \alpha \beta \right) \\
&= \inf_{\tau} \left( C( \tau ) \geq 1 - \alpha \beta \right) \\
&\leq \inf_{\tau} \left( F( \tau ) \geq 1 - \alpha \beta \right) \\
& \; \; \mbox{ (since the extremal conditional marginal bounds the unconditioned marginal) } \\
&= \inf_{\tau} \left( \tau \geq F^{-1}(1 - \alpha \beta) \right) \\
&= F^{-1}(1 - \alpha \beta)
\end{align*}

Thus under our assumptions, $\tau^1_{\alpha, \beta} \leq F^{-1}(1 - \alpha \beta)$.

We now bound $\tau^2_{\alpha}$.
\begin{align*}
\tau_{\alpha}^2 &= \inf_{\tau} \left(\min_{x_i,y_i} \left( Pr_{c_j \sim \mathcal{C}|x_i, y_i} [c_j \leq \tau] \right) \geq 1 - \alpha \right) \\
&\leq \inf_{\tau} \left( F(\tau) \geq 1 - \alpha \right) \\
&= \inf_{\tau} \left( \tau \geq F^{-1}(1 - \alpha) \right) \\
&= F^{-1}(1 - \alpha)
\end{align*}

Finally, 
\[ \tau_{\alpha, \beta} = \max \{ \tau_{\alpha, \beta}^1, \tau_{\alpha}^2 \} \leq \max\{ F^{-1}(1 - \alpha \beta) , F^{-1}(1 - \alpha) \}.\]
\end{proof}

\subsection{Proof of Corollary \ref{cor.formal} (Main result)}

\begin{cor*}
Choose $\delta \in (0, \frac{p}{2+2p})$.  Then under Assumptions \ref{a.costs}, \ref{a.indep}, and \ref{a.tail}, setting $\alpha = n^{-\delta}$, $\beta = n^{-\frac{p}{2} + \delta(1+p)}$, $\epsilon = n^{-1 + \delta}$, $\gamma = n^{1 - \frac{\delta}{2}}$, $a = (6B+2M)(1+B)^2 n^{-\frac{3}{2}} + n^{-\frac{3}{2} + \delta}$, $b = n^{-\frac{3}{2}}$, $\xi = 1/2$,  and $t = \sqrt{\frac{n}{4C(d+2)\log d}}$ in Algorithm \ref{alg:private} ensures that with probability $1 - d^{\Theta\left(-n \right)} - n^{-\frac{p}{2} + \delta(1+p)}$:
\begin{enumerate}
\item the output of Algorithm \ref{alg:private} is $O\left(n^{-1 + \delta}\right)$-jointly differentially private, 
\item it is an $O\left(n^{-\frac{3}{2} + \delta }\right)$-approximate Bayes Nash equilibrium for a $1-O\left( n^{-\delta} \right)$ fraction of players to truthfully report their data,
\item the computed estimate $\pr$ is $O \left( n^{-\delta} \right)$-accurate,
\item it is individually rational for a $1-O\left( n^{-\delta} \right)$ fraction of players to  participate in the mechanism, and
\item the required budget from the analyst is $O\left(n^{-\frac{1}{2} + \delta}\right)$.
\end{enumerate}
\end{cor*}

\begin{proof}
Choose $\delta \in (0, \frac{p}{2+2p})$.  Note that this ensures $\delta < 1/2$.  Let $\alpha = n^{-\delta}$ and $\beta = n^{\frac{p}{2} - \delta(1+p)}$ as we have chosen.  By the constraint that $\delta < \frac{p}{2+2p}$, we have ensured that $\beta = o(1)$. By Lemma \ref{lem.tail}, $\tau_{\alpha, \beta} \leq \max\{(\alpha \beta)^{-1/p}, \alpha^{-1/p} \} = (\alpha \beta)^{-1/p}$ since $\alpha, \beta = o(1)$ and $p>1$.  Then $\tau_{\alpha, \beta} = O\left(n^{1-\delta}\right)$.

Setting $\xi = 1/2$ and $t = \sqrt{\frac{n}{4C(d+2)\log d}}$, we ensure that with probability $1 - d^{-\frac{n}{4C(d+2)\log d}} = 1 - d^{\Theta(-n)}$, the bounds stated in Theorem \ref{thm.converge} hold.  With probability $1-\beta$, at most an $\alpha$-fraction of players will have cost parameters above $\tau_{\alpha, \beta}$.  Taking a union bound over these two failure probabilities, the bounds in Theorems \ref{thm.priv}, \ref{thm.equil}, \ref{thm.accuracy}, \ref{thm.privateir}, and  \ref{thm.budget} will all hold with probability at least $1 - d^{\Theta(-n)} - n^{-\frac{p}{2} + \delta(1+p)}$.  For the remainder of the proof, we will assume all bounds hold, which will happen with at least the probability specified above.

First note that by Theorem \ref{thm.priv}, Algorithm \ref{alg:private} is $2 \eps$-jointly differentially private.  By our choice of $\epsilon$, the privacy guarantee is $2 n^{-1 + \delta} = o(\sqrt{n})$.

Recall that by Theorem \ref{thm.equil}, it is a $\left[b \left( \frac{\alpha n}{\gamma}(4B+2M) + \frac{\gamma B}{\gamma + (1-\xi) \frac{1}{d+2} n} \right)^2 + \tau_{\alpha, \beta} \eps^2\right]$-approximate Bayes-Nash equilibrium for a $(1 - \alpha)$-fraction of players to truthfully report their data.  Taking $B$, $M$, $\xi$, and $d$ to be constants, it is a $\Theta\left(b \left( \frac{\alpha n}{\gamma} + \frac{\gamma}{n} \right)^2 + \tau_{\alpha, \beta} \eps^2 \right)$-approximate BNE.  To achieve the desired truthfulness bound, we require (among other things) that $\tau_{\alpha, \beta} \eps^2 = o(\frac{1}{n})$.  Given the bound on $\tau_{\alpha, \beta}$, it would suffice to have $\eps = o(n^{-\frac{3}{4} + \frac{\delta}{2}})$.  This is satisfied by our choice of $\eps = n^{-1 + \delta}$ because $\delta <1/2$.  After setting $b = o(\frac{1}{n})$, we will have the desired truthfulness bound if $\frac{\alpha n}{\gamma} + \frac{\gamma}{\gamma + n} = o(1)$.  This implies the following constraints on $\gamma$: we require $\gamma = \omega(n \alpha) = \omega(n^{1 - \delta})$ and $\gamma = o(n)$.  Our choice of $\gamma = n^{1 - \frac{\delta}{2}}$ satisfies these requirements.  Due to our choice of $b = n^{-3/2}$, the approximation factor will be dominated by $\tau_{\alpha, \beta} \eps^2 = O\left(n^{-\frac{3}{2} + \delta}\right) = o(1)$.  Thus truthtelling is an $O\left(n^{-\frac{3}{2} + \delta}\right) = o(1)$-approximate Bayes-Nash equilibrium for all but an $n^{-\delta} = o(1)$-fraction of players.

Recall from Theorem \ref{thm.accuracy} that the estimator $\pr$ is $O\left( \left(\frac{\alpha n}{\gamma} + \frac{1}{\gamma \eps}\right)^2 + \left(\frac{\gamma}{\gamma + n}\right)^2 + \left(\frac{1}{n}\right)^2 + \frac{\alpha n}{\gamma} + \frac{1}{\gamma \eps} \right)$-accurate.  We have already established that $\frac{\alpha n}{\gamma} = o(1)$ and $\frac{\gamma}{\gamma + n} = o(1)$.  Trivially, $\frac{1}{n^2} = o(1)$.  We turn now to the term $\frac{1}{\gamma \eps}$.  For this term to be $o(1)$, we require $\gamma = \omega(\frac{1}{\eps}) = \omega\left(n^{1 - \delta}\right)$.  Our choice of $\gamma = n^{1 - \frac{\delta}{2}}$ ensures this requirement is satisfied. Since $\frac{\alpha n}{\gamma} + \frac{1}{\gamma \eps} = o(1)$, then so must be $\left(\frac{\alpha n}{\gamma} + \frac{1}{\gamma \eps}\right)^2 = o(1)$.  The accuracy bound will be dominated by three terms: first $\left(\frac{\gamma}{n}\right)^2 = n^{-\delta}$, second $\frac{\alpha n}{\gamma} = n^{-\frac{\delta}{2}}$, and third $\frac{1}{\gamma \eps} = n^{-\frac{\delta}{2}}$.  Thus, Algorithm \ref{alg:private} outputs an estimator with accuracy $O\left(n^{-\frac{\delta}{2}} \right) = o(1)$.

Theorem \ref{thm.privateir} says that the mechanism in Algorithm \ref{alg:private} is individually rational for a $(1-\alpha)$-fraction of players as long as $a \geq \left( \frac{\alpha n}{\gamma}(4B+2M) + \frac{\gamma B}{\gamma + (1-\xi) \frac{1}{d+2} n} + B \right) (b+2bB) + bB^2 + \tau_{\alpha, \beta} \eps^2$.  We now expand each term of this expression to prove that our choice of $a$ satisfies the desired bound. Consider the first term: $\frac{\alpha n}{\gamma}(4B+2M) = n^{-\frac{\delta}{2}}(4B+2M)$.  This term is decreasing in $n$, so it can be upper bounded by its value when $n=1$.  Thus $\frac{\alpha n}{\gamma}(4B+2M) \leq 4B+2M$.  Now consider the second term:
\[ \frac{\gamma B}{\gamma + (1-\xi) \frac{1}{d+2} n} = \frac{n^{1 - \frac{\delta}{2}} B}{n^{1 - \frac{\delta}{2}} + \frac{1}{2(d+2)} n} = \frac{n^{- \frac{\delta}{2}} B}{n^{- \frac{\delta}{2}} + \frac{1}{2(d+2)}} = B\left(1 - \frac{1}{2(d+2) n^{-\frac{\delta}{2}} + 1} \right) \]
The final term $\frac{-1}{2(d+2) n^{-\frac{\delta}{2}} + 1}$ is always negative, so the entire term $\frac{\gamma B}{\gamma + (1-\xi) \frac{1}{d+2} n}$ can be bounded above by $B$.  We can simplify the expression $b+2bB+bB^2$ as $(1+B)^2 b = (1+B)^2 n^{-3/2}$.  Finally, as noted earlier (and due to to Lemma \ref{lem.tail}), we can upper bound $\tau_{\alpha, \beta} \eps^2 \leq n^{-\frac{3}{2} + \delta}$.  Combining all of these bounds, it would suffice to set $a \geq (6B+2M)(1+B)^2 n^{-3/2} + n^{-\frac{3}{2} + \delta}$.  We set $a$ to be exactly this bound.  Then it is individually rational for a $1- \alpha = 1 - n^{-\delta} = 1 - o(1)$ fraction of players to participate in the mechanism.

By Theorem \ref{thm.budget}, the budget required from the analyst is $\mathcal{B} \leq n \left[a + \left( \frac{\alpha n}{\gamma}(4B+2M) + \frac{\gamma B}{\gamma + (1-\xi) \frac{1}{d+2} n} + B \right) (b+2bB) \right]$.  From our choice of $a=\Theta\left(n^{-\frac{3}{2} + \delta}\right)$ and because $\frac{\alpha n}{\gamma} + \frac{\gamma}{n} = o(1)$, the required budget is $\mathcal{B} = O\left(n(b + \tau_{\alpha, \beta}\eps^2)\right) = O\left(n(n^{-\frac{3}{2}} + n^{-\frac{3}{2} + \delta})\right) = O\left(n^{-\frac{1}{2} + \delta}\right) = o(1)$.
\end{proof}


%% file: appendix.tex
\section{Proof of Theorem \ref{thm.converge}}\label{app:vers}

\converge*

\begin{proof}
We will first require Lemma \ref{lem.covariance}, which characterizes the covariance matrix of the distribution on $X$.  

\begin{lemma}\label{lem.covariance}
The covariance matrix of $x$ is $\Sigma = \frac{1}{d+2} I$.
\end{lemma}

\begin{proof}
Let $z_1, \ldots, z_d \sim N(0,1)$, and let $u \sim U[0,1]$, all drawn independently.  Define, $r = \sqrt{z_1^2 + \cdots + z_d^2}$ and $Z = (u^{1/d} \frac{z_1}{r}, \ldots, u^{1/d} \frac{z_d}{r})$.  Then $Z$ describes a uniform distribution over the $d$-dimensional unit ball \citep{Knu81}.  Recall that this is the same distribution from which the $x_i$ are drawn.  By the symmetry of the uniform distribution, $\E[Z] = \vec{0}$, and $Cov(Z)$ must be some scalar times the Identity matrix.  Then to compute the covariance matrix of $Z$, it will suffice to compute the variance of some coordinate $Z_i$ of $Z$.  Since each coordinate of $Z$ has mean 0, then $Var(Z_i) = \E[Z_i^2] + \E[Z_i]^2 = \E[Z_i^2]$.
\begin{eqnarray*}
\sum_{i=1}^d \E[Z_i^2] &=& \E \left[ \sum_{i=1}^d Z_i^2 \right] \\
&=& \E \left[ \sum_{i=1}^d \left(u^{1/d} \frac{z_i}{r} \right)^2 \right] \\
&=& \E[u^{2/d}] \E \left[ (\frac{1}{r})^2 \sum_{i=1}^d z_i^2 \right] \\
&=& \E[u^{2/d}] \\
&=& \frac{d}{d+2}
\end{eqnarray*}
By symmetry of coordinates, $\E[Z_i^2] = \E[Z_j^2]$ for all $i,j$. Then $\E[Z_i^2] = \frac{1}{d+2}$, and the covariance matrix of $Z$ (and of the $x_i$ since both variables have the same distribution) is $\Sigma = \frac{1}{d+2} I$.
\end{proof}

From Corollary 5.52 in \cite{Ver11} and the calculation of covariance in Lemma \ref{lem.covariance}, for any $\xi \in (0,1)$ and $t \geq 1$, with probability at least $1-d^{-t^2}$, 
\begin{equation}\label{eq:ver} \left\| \frac{1}{n} X^{\top} X - \frac{1}{d+2}I \right\| \leq \xi \frac{1}{d+2}, \end{equation}
when $n \geq C(\frac{t}{\xi})^2 (d+2) \log d$, for some absolute constant $C$.  We assume for the remainder of the proof that inequality \eqref{eq:ver} holds, which is the case except with probability at most $d^{-t^2}$, as long as $n$ is sufficiently large. Then 
\[ \left\| X^{\top} X - \frac{1}{d+2}nI \right\| \leq \xi \frac{1}{d+2}n. \]

Let $\lambda_{\max}(A)$ and $\lambda_{\min}(A)$ denote respectively the maximum and minimum eigenvalues of a matrix $A$. By definition, $\lambda_{\max}(A) = \| A \|$.

Assume towards a contradiction that $\lambda_{\max}(X^{\top} X) = (1+\xi)\frac{1}{d+2}n + \delta$ for $\delta > 0$.
\begin{eqnarray*}
\xi \frac{1}{d+2}n & \geq & \left\| X^{\top} X - \frac{1}{d+2} n I \right\| \\
&=& \left\| X^{\top} X \right\| - \frac{1}{d+2} n  \\
&=& \lambda_{\max}(X^{\top} X) - \frac{1}{d+2} n \\
&=& (1+\xi)\frac{1}{d+2}n + \delta - \frac{1}{d+2} n \\
&=& \xi \frac{1}{d+2}n + \delta
\end{eqnarray*}
This implies $\delta \leq 0$, which is a contradiction.  Thus $\lambda_{\max}(X^{\top} X) = \| X^{\top} X \| \leq (1+\xi)\frac{1}{d+2}n$.

Similarly, assume that $\lambda_{\min}(X^{\top} X) = (1-\xi)\frac{1}{d+2}n - \delta$ for some $\delta > 0$.  Since all eigenvalues are positive, it must be the case that $\lambda_{\min}(X^{\top} X) \geq 0$.
\begin{eqnarray*}
0 & \geq & \lambda_{\min}(X^{\top} X - \frac{1}{d+2} n I) \\
&=& \lambda_{\min}(X^{\top} X) - \frac{1}{d+2} n \\
&=& (1-\xi)\frac{1}{d+2}n - \delta - \frac{1}{d+2} n \\
&=& -\xi \frac{1}{d+2}n - \delta
\end{eqnarray*}
This is also a contradiction, so $\lambda_{\min}(X^{\top} X) \geq (1-\xi)\frac{1}{d+2}n$.  For any matrix $A$, $\lambda_{\max}(A^{-1}) = \frac{1}{\lambda_{\min}(A)}$. Thus,
\begin{eqnarray*} 
\lambda_{\min}(X^{\top} X) &=& \frac{1}{\lambda_{\max}\left( (X^{\top} X)^{-1} \right)} \\
&=& \frac{1}{\| (X^{\top} X)^{-1} \|} \\
&\geq& (1-\xi)\frac{1}{d+2}n \\
\Longrightarrow \; \; \; \; \| (X^{\top} X)^{-1} \| &\leq& (1-\xi)\frac{1}{d+2}n
\end{eqnarray*}

Using the fact that $\lambda$ is an eigenvalue of a matrix $A$ if and only if $(\lambda + c)$ is an eigenvalue of $(A + cI)$, we have the following inequalities to complete the proof:
\[\left\| \gamma I + X^{\top} X \right\| = \lambda_{\max}(\gamma I + X^{\top} X) \leq \gamma + (1+\xi)\frac{1}{d+2} n \]
\[\left\| (\gamma I + X^{\top} X)^{-1} \right\| = \frac{1}{\lambda_{\min}(\gamma I + X^{\top} X)} \leq \frac{1}{\gamma + (1-\xi)\frac{1}{d+2} n} \]

\end{proof}

%% file: appendix2.tex
\section{Quadratically Bounded Privacy Penalty Costs}\label{app:quad}

We will consider a particular functional form of $f_i(c_i, \eps)$, motivated by the model of privacy cost in the existing literature \citep{Harvardetal}.  In particular, we assume that each player additionally has a privacy cost function $g_i$ that measures her loss for participating in a particular instantiation of a mechanism.  Further, we assume that $g_i$ is upper-bounded by a function that depends on the effect that player $i$'s report has on the mechanism's output.  This assumption leverages the functional relationship between player $i$'s data $(x_i, y_i)$, and the output of the mechanism.  For example, if a particular mechanism ignores the input from player $i$, then her privacy cost should be 0 for participating in that computation, since her data is not used.  We then define her ex ante privacy cost $f_i(c_i, \eps)$ to be her expected cost for participation, where the expectation is taken over the randomness of other players' data and reports.

To formally state this assumption, first let mechanism $\cM$ take in data reports $(X,y)$ and output an estimated parameter $\hatth$.  Define $g_i(M, \hatth, (x_i, y_i), (X_{-i},y_{-i}))$ to be the privacy cost to player $i$ for reporting $(x_i, y_i)$ to mechanism $\cM$ when all other players report $(X_{-i}, y_{-i})$ and the output of $\cM$ is $\hatth$.

\begin{assumption}[\cite{Harvardetal}, Privacy Cost Assumption]\hspace{-1.5mm}\footnote{The assumption proposed in \cite{Harvardetal} allows privacy costs to be bounded by an arbitrary function of the log probability ratio that satisfies certain natural properties.  We restrict to this particular functional form for simplicity, following \cite{GLRS14}.}\label{a.cost}
We assume that for any mechanism $M$ that takes in data $(X,y)$ and outputs an estimate $\hatth$, then for all players $i$, for all estimates $\hatth$, and for all possible input data $(X,y)$,
\[ g_i(M, \hatth, (x_i, y_i), (X_{-i},y_{-i})) \leq c_i \ln \left( \max_{y_i',y_i''} \frac{Pr[M(X,y_i',y_{-i}) = \hatth]}{Pr[M(X,y_i'',y_{-i}) = \hatth]} \right). \]
\end{assumption}

\begin{lemma}[\cite{DRV10, Harvardetal}, Composition Lemma]\label{lem.comp}
In settings that satisfy Assumption \ref{a.cost} and for mechanisms $M$ that are $\eps$-differentially private for $\eps \leq 1$, then for all players $i$ with data $(x_i,y_i)$, for all data reports of other players $(X_{-i}, y_{-i})$, and for all possible misreports $y_i'$ by player $i$,
\[ \E[g_i(M, M(X,y), (x_i, y_i), (X_{-i},y_{-i}))] - \E[g_i(M, M(X, y_i', y_{-i}), (x_i, y_i), (X_{-i},y_{-i}))] \leq 2 c_i \eps(e^{\eps} - 1) \leq 4 c_i \eps^2 \]
\end{lemma}
\begin{proof} (Sketch)
The first inequality comes from Lemma 5.2 of \cite{Harvardetal} by plugging in our specification of their ``privacy-bound function'' and replacing statistical difference with the upper bound of $e^{\eps} - 1$.  The second inequality comes from the bound $e^{\eps} \leq 1 + 2\eps$ for small $\eps$.
\end{proof}

To combine this framework with the utility model introduced in Section \ref{s.costs}, we need only to interpret $f_i(c_i, \eps)=\frac{1}{4}\E[g_i(M, M(X,y), (x_i, y_i), (X_{-i},y_{-i}))]$.  That is, $f(c_i, \eps)$ is player $i$'s expected cost for participating in the mechanism (up to a scaling constant).  This interpretation, along with Lemma \ref{lem.comp}, motivates Assumption \ref{a.costs}.

\section{Strong Convexity of Regularized Loss}\label{app:strongconv}
Recall that we consider the loss function $\cL(\theta, X,y)$ to be the sum of these individual loss functions plus a regularizing term:
\[ \cL(\theta; X,y) = \sum_{i=1}^n \ell(\theta; x_i, y_i) = \sum_{i=1}^n (y_i - \theta^{\top} x_i)^2 + \gamma \left\| \theta \right\|_2^2. \]

We now define strong convexity, which requires that the eigenvalues of the Hessian of a function are bounded away from zero, and we prove that the loss function $\cL$ is strongly convex.
\begin{definition}[Strong Convexity]
A function $f: \R^d \to \R$ is $m$-strongly convex if
\[ H\left(f(\chi)\right) - m I \mbox{ is positive semi-definite for all } \chi \in \R^d, \]
where $H(f(\chi))$ is the Hessian\footnote{The \emph{Hessian} $H$ of function $f$ is a $d \times d$ matrix of its partial second derivatives, where 
\[ H(f(\chi))_{jk} = \frac{\partial^2 f(\chi)}{\partial \chi_j \partial \chi_k}. \]
A $d \times d$ matrix $A$ is \emph{positive semi-definite} (PSD) if for all $v \in \R^d$, $v^{\top}Av \geq 0$.} of $f$, and $I$ is the $d \times d$ identity matrix.
\end{definition}

Notice that when $f$ is a one-dimensional function ($d=1$), strong convexity reduces to the requirement that $f''(\chi) \geq m > 0$ for all $\chi \in \R$.  The following lemma proves that regularizing the quadratic loss $\cL$ ensures that it is strongly convex.
\begin{lemma}\label{lem.strongconv}
$\cL(\theta; X,y)$ is $2\gamma$-strongly convex in $\theta$.
\end{lemma}


\begin{proof}
We first compute the Hessian of $\cL(\theta; X,y)$.  For notational ease, we will suppress the dependence of $\cL$ on $X$ and $y$, and denote the loss function as $\cL(\theta)$.  We will use $x_{ij}$ to denote the $j$-th coordinate of $x_i$, and $\theta_j$ to denote the $j$-th coordinate of $\theta$.
\begin{eqnarray*} \frac{\partial \cL(\theta)}{\partial \theta_j} &=& \sum_{i=1}^n \left[ -2 y_i x_{ij} + 2(\theta^{\top} x_i) x_{ij} \right] + 2 \gamma \theta_j \\
\frac{\partial \cL(\theta)}{\partial \theta_j \partial \theta_k} &=&  \sum_{i=1}^n \left[ 2(x_{ik}) x_{ij} \right] \mbox{ for } j \neq k \\
\frac{\partial \cL(\theta)}{\partial \theta_j^2} &=&  \sum_{i=1}^n \left[ 2(x_{ij})^2 \right] + 2\gamma 
\end{eqnarray*}
The Hessian of $\cL$ is, 
\[ H(\cL(\theta)) = \sum_{i=1}^n x_i x_i^{\top} + 2 \gamma I, \]
where $I$ is the identity matrix.  Thus, 
\[ H(\cL(\theta)) - 2 \gamma I = \sum_{i=1}^n x_i x_i^{\top}, \] 
which is positive semi-definite.  To see this, let $v$ be an arbitrary vector in $\R^d$.  Then for each $i$, $v(x_i x_i^{\top})v^{\top} = (v x_i)^2 \geq 0$. The sum of PSD matrices is also PSD, so $\cL(\theta)$ is $2 \gamma$-strongly convex.
\end{proof}